%% file: Main.tex
\documentclass[11pt]{article}
\usepackage[colorlinks=true, linkcolor=blue, citecolor=blue, backref=page]{hyperref}

\usepackage[letterpaper, margin=1.1in]{geometry}

\usepackage{amsmath}
\usepackage{amsthm}
\usepackage{amsfonts}
\usepackage{amssymb}
\usepackage{mathtools}
\usepackage{enumitem}
\usepackage{microtype}
\usepackage[vlined, ruled, algo2e, algosection]{algorithm2e}
\usepackage[displaymath, mathlines, modulo]{lineno}
\usepackage{hyperref}
\usepackage{subcaption}
\usepackage{floatrow}
\newfloatcommand{capbtabbox}{table}[][\FBwidth]

\usepackage{thm-restate}
\usepackage{booktabs}
\usepackage{tikz}
\usetikzlibrary{arrows, topaths, calc, trees, bending, positioning}

\renewcommand*\backref[1]{\ifx#1\relax \else ($\uparrow$ #1) \fi}

\renewcommand{\setminus}{\backslash}

\DeclareMathOperator*{\argmax}{arg\,max}

\newcommand{\card}[1]{\left|#1\right|}
\newcommand{\claw}[1]{C_{#1}}
\newcommand{\cC}{\mathcal{C}}

\newcommand{\Nap}{N_a^+}
\newcommand{\Napm}{N_a^+ \backslash\{a\}}

\newtheorem{theorem}{Theorem}[section]
\newtheorem*{theorem*}{Theorem}
\newtheorem{lemma}[theorem]{Lemma}
\newtheorem{corollary}[theorem]{Corollary}

\theoremstyle{definition}
\newtheorem{Definition}[theorem]{Definition}
\theoremstyle{remark}

\newcommand{\opt}{\mathrm{OPT}}

\renewcommand{\epsilon}{\varepsilon}
\newcommand{\e}{\varepsilon}
\newcommand{\vmax}{\hat{v}}
\newcommand{\ierat}{1 - \tfrac{\epsilon - \delta}{1-\epsilon}}
\newcommand{\ieratsimp}{2 - \tfrac{1}{\sqrt{1-\epsilon}}}

\begin{document}

\title{An Improved Approximation for Maximum Weighted $k$-Set Packing}
\author{Theophile Thiery\thanks{School of Mathematical Sciences, Queen Mary University of London, London, United Kingdom (\href{mailto:t.f.thiery@qmul.ac.uk}{t.f.thiery@qmul.ac.uk},
    \href{mailto:justin.ward@qmul.ac.uk}{justin.ward@qmul.ac.uk}). This work was supported by the Engineering and Physical Sciences Research Council [EP/T006781/1].} \and Justin Ward\footnotemark[1]}
\date{\today}

\maketitle







\begin{abstract}
We consider the weighted $k$-set packing problem, in which we are given a collection of weighted sets, each with at most $k$ elements and must return a collection of pairwise disjoint sets with maximum total weight. For $k = 3$, this problem generalizes the classical 3-dimensional matching problem listed as one of the Karp's original 21 NP-complete problems. We give an algorithm attaining an  approximation factor of $1.786$ for weighted 3-set packing, improving on the recent best result of $2-\frac{1}{63,700,992}$ due to Neuwohner.

Our algorithm is based on the local search procedure of Berman that attempts to improve the sum of squared weights rather than the problem's objective. When using exchanges of size at most $k$, this algorithm attains an approximation factor of $\frac{k+1}{2}$. Using exchanges of size $k^2(k-1) + k$, we provide a relatively simple analysis to obtain an approximation factor of 1.811  when $k = 3$. We then show that the tools we develop can be adapted to larger exchanges of size $2k^2(k-1) + k$ to give an approximation factor of 1.786. Although our primary focus is on the case $k = 3$, our approach in fact gives slightly stronger improvements on the factor $\frac{k+1}{2}$ for all $k > 3$. As in previous works, our guarantees hold also for the more general problem of finding a maximum weight independent set in a $(k+1)$-claw free graph.
\end{abstract}

\section{Introduction}
In the weighted $k$-set packing problem, we are given a weighted collection of of $n$ sets, each containing at most $k$ elements from some universe. The goal is to return a collection of disjoint sets of maximum total weight. The weighted $k$-set packing problem generalizes many practical and theoretical problems. When $k=2$, it encompasses the maximum weight matching problem. For larger $k$, it generalizes the $k$-dimensional matching problem, which involves finding a maximum matching in a $k$-partite $k$-uniform hypergraph.
While the maximum weight matching problem admits a polynomial time algorithm which solves the problem exactly~\cite{Edmonds:1965:Paths}, the 3-dimensional matching problem is NP-hard even in the unweighted case~\cite{DBLP:conf/coco/Karp72}. For low values of $k$, unweighted $k$-dimensional matching is in fact NP-hard even to approximate beyond a factor of $98/97$, $54/53$, $30/29$ and $23/22$ for $k = 3,4,5,$ and $6$, respectively~\cite{DBLP:journals/eccc/ECCC-TR03-008,Hazan:2006:Complexity}, and $\Omega(k/\ln(k))$ for general $k$~\cite{Hazan:2006:Complexity}.\footnote{To remain consistent with previous work, of our approximation results are stated using using the approximation factor $w(O)/w(A) \geq 1$, which measures how much better the optimal solution $O$ is compared to the solution $A$ produced by the algorithm.}

In contrast, the best approximation algorithm for the unweighted problem is a $\frac{k+1+\e}{3}$-approximation due to Cygan~\cite{Cygan:2013:Improved} with subsequent improvements by F\"urer and Yu~\cite{Furer:2014:Approximating} to the running time dependence on $\e$. It is instructive to observe that all of the best known algorithms in the unweighted regime use local-search procedures, which repeatedly improve a solution $S$ by repeatedly adding some small number of sets not currently in $S$ and removing  intersecting sets from $S$. If each such swap attempts to add only one set at a time, then this leads to a $k$-approximation. Hurkens and Schrijver~\cite{Hurkens:1989:Size} showed that for any $\epsilon > 0$, an algorithm performing swaps of size $O(\epsilon^{-1})$ gives a $\frac{k+\epsilon}{2}$-approximation, and subsequent improvements to $\frac{k+1+\e}{3}$~\cite{Cygan:2013:Improved,Furer:2014:Approximating} have been obtained by increasing the swap size further to $\Omega(\log(n))$.

Surprisingly, in the case of \emph{weighted} $k$-set packing, using swaps of size $O(\epsilon^{-1})$ leads to an approximation factor of only $k - 1 + \e$~\cite{Arkin:1998:Local}. However, Berman~\cite{Berman:2000:d/2} showed that by \emph{squaring} the weight of each set and using swaps of size $k$ to find a local optimum of the resulting instance results in a $\frac{k+1+\e}{2}$-approximation with respect to the original weight function (where here the $\e$ is due to a further rescaling procedure to ensure the algorithm terminates in polynomial time). 

Berman's algorithm in fact applies to the more general problem of finding a \emph{maximum weight independent set of vertices in a $(k+1)$-claw free graph}. Briefly, a $d$-claw is an induced subgraph of $G$ comprising a single vertex (called the \emph{center} of the claw) adjacent to a set of $d$ pairwise non-adjacent vertices (called the \emph{talons} of the claw). A graph is then $(k+1)$-claw free if it contains no $(k+1)$-claw. By creating a graph containing a vertex for each set in a $k$-set packing instance and an edge between sets that are non-disjoint, we can convert the (weighted) set packing problem to an instance of the (weighted) independent set problem, and if each set has size at most $k$, then the maximum size of a claw in the resulting graph is also $k$. We call the graph $G$ obtained in this way the \emph{conflict graph} for the underlying set packing instance. For simplicity, we will henceforth consider the general problem of finding a maximum weight independent set in some vertex weighted $(k+1)$-claw free graph and adopt the associated vocabulary.

In this vocabulary, Berman's local search algorithm squares the weight of all vertices of the graph and then considers a restricted set of ``claw swaps.'' For each vertex $a$ in some current solution $A$, the algorithm searches for a claw of $G$ centered at $a$ and then adds the talons of this claw to $A$ and discards any conflicting vertices from $A$ as long as this increases the total (now squared) weight of $A$. The key difficulty in the analysis of the algorithm is in translating local optimality with respect to the squared weighting function $w^2$ into a guarantee in terms of the original weight function $w$. To accomplish this, Berman employs a 2-round charging argument, whereby vertices in the optimal solution distribute their weight among neighboring vertices in the locally optimal solution $A$ produced by the algorithm.

For over 20 years, Berman's algorithm has remained the state-of-the art approximation result for both weighted $k$-set packing and maximum weight independent set in $(k+1)$-claw free graphs. In a recent breakthrough result, Neuwohner~\cite{Neuwohner:2021:Improved} was able to break the barrier of $\frac{k+1}{2}$ and obtain a slightly improved approximation ratio equal to $\frac{k+1+\e}{2}- \frac{1}{63,700,992}$ by squaring the weights and then considering larger exchanges than in Berman's algorithm. 
The key observation behind their analysis is that the charging argument employed by Berman is only tight when the weights of  vertices in $A$ and $O$ are nearly identical. 
Neuwohner's analysis leverages this observation to create a more complex charging scheme that considers several different classes of vertices. She then argues that in any solution $A$ that is locally optimal under swaps of size $O(k^2)$, there must exist some set of vertices with weight constituting a significant fraction of the weight $A$, which receive less than $\frac{k+1}{2}$ times their weight under the new charging scheme. 
In a follow-up paper, Neuwohner~\cite{Neuwohner:2022:Limits} showed that for any local-search algorithm that works by improving some power $w^\alpha$ of the weights cannot improve on the factor $\frac{k}{2}$ even using swaps of size $O(\log n)$. However, Neuwohner manages to attain the factor $\frac{k}{2}$ asymptotically using swaps of size $O(\log n)$. She proves that for any $\delta>0$, there is a $k_\delta$ such that for any $k \geq k_\delta$, considering swaps of size $O(\log n)$ with the squared weighting has approximation ratio $\frac{k+\delta}{2}$. The threshold is equal to $k_\delta = \frac{200,000}{\delta^3}$. As the rate of convergence to $\frac{k}{2}$ is relatively slow, for $k=3$, where the potential for improvement in the ratio is the largest, the best factor remains $\frac{k+1+\e}{2}- \frac{1}{63,700,992}$. In further work, Neuwohner~\cite{DBLP:journals/corr/abs-2202-01248} has recently shown that the barrier of $\frac{k}{2}$ can in fact be surpassed by running the unweighted local search algorithm on appropriate sub-instances of a given instance. The techniques she employs require that $k \geq 4$. When $k=4$, she obtains an improvement of $0.002$ over the factor of $\frac{4+1}{2}$. As with previous results, the improvement over the factor $\frac{k+1}{2}$ grows with $k$ to $0.0115$ when $k = 13$, and $0.4986(k+1) + 0.0208$ for all $k \geq 14$.

\subsection*{Our Results}
\label{sec:our-results}
Given this stream of recent progress in the asymptotic approximability of the weighted $k$-set packing for large $k$, it is natural to ask whether it is possible to obtain significant improvements in approximation specifically in the case of small $k$.
Here we answer this question affirmatively, by giving two new approximation guarantees for the weighted $k$-set packing problem by using a variant of Berman's squared-weight local search with larger exchanges. We first present a relatively simple analysis showing that exchanges adding up to $k^2(k-1) + k$ sets is sufficient to obtain a factor 1.811 for weighted 3-set packing, improving on the factor $\frac{k+1}{2} = 2$ by 0.189. We then show that by refining our basic analysis, it is possible to attain a 1.786-approximation using swaps of size $2k^2(k-1) + k$. Our results imply better improvements for $k > 3$, and we show that our algorithms' guarantees improve asymptotically to $(k+\frac{1}{2})/2$ and $(k+\frac{1}{3})/2$, respectively, as $k$ grows. We summarize our results in the following theorem.

\begin{theorem}[Summary]
\label{thm:main}
A squared-weight local search algorithm performing exchanges of size $2k^2(k-1) + k$ is a polynomial time $\frac{k+1-\tau_k}{2}$-approximation for the weighted $k$-set packing problem, where $\tau_k \geq \tau_3 = 0.214$ and $\lim_{k\to\infty} \tau_k = 2/3$. The same algorithm with exchanges of size $k^2(k-1) + k$, is a $\frac{k+1-\tau'_k}{2}$-approximation with $\tau'_k \geq \tau'_3 = 0.189$ and $\lim_{k \to \infty} \tau'_k = 1/2$.
\end{theorem}

Further specific values for our approximation guarantee, as well as the improvement $\tau_k/2$, $\tau'_k/2$ that we make over $\frac{k+1}{2}$ are given in Table~\ref{tab:ratio}. The precise value of $\tau_k$ depends on considering and balancing the worst of several quantities. To provide a brief overview, here we have simply listed the final results that follow from our techniques. After performing our main analysis, we provide and prove a more detailed version of Theorem~\ref{thm:main} that explains how the numerical quantities in Table~\ref{tab:ratio} were obtained.  This is given in  Theorem~\ref{thm:bound-2}.

While our results are also based on considering larger exchanges in the squared-weight local search algorithm introduced by Berman~\cite{Berman:2000:d/2}, we adopt a different approach than that employed by Neuwohner~\cite{Neuwohner:2022:Limits}. We first (in Section~\ref{sec:berman}) give a compact proof of Berman's guarantee that avoids an explicit charging argument. This allows us to make explicit the slack present in the technical inequalities used to relate $w^2$ to $w$ using local optimality. For each vertex $a$ in a locally optimal solution $A$, we consider two different types of slack. The first, which we denote by $\Delta_a$, captures the tightness of the claw swap centered $a$ (i.e.\ how much the total squared weight of $A$ would decrease after performing the claw-swap centered at $a$). The second term, which we denote as $\Psi_a$, measures the slack in the remaining argument due to the deviation of the weight of the talons and of the neighbors of the talons from the weight of $a$. More precisely, $\Psi_a$ captures the slack in two technical inequalities applied to vertex weights: $xy \leq \frac{1}{2}x^2 + \frac{1}{2}y^2$ and $\sum_{i}z_i^2 \leq (\max_{i}z_i)\sum_{i}z_i$, where all $z_i > 0$. Both of these inequalities are tight only when $x = y$ and all $z_i$ are equal, respectively.

Our analysis then works by considering exchanges of size $O(k^3)$ and bounding the sum of $\Delta_a$ and $\Psi_a$ away from $0$ in two cases. In the first case, suppose that a vertex $a \in A$ has some vertex $b$ of similar weight that would be removed by the claw swap centered at $a$. Then, we show that the swaps centered at $a$ and $b$ cannot both be tight, since otherwise the swap which brings the sets of talons of $a$ and $b$ together would be improving. Hence, for any vertex $a$ with a ``close'' vertex $b$ of this sort, $\Delta_a + \Delta_b$ must be bounded away from $0$, where the exact amount depends on the similarity between $a$'s weight and $b$'s weight. In order to exploit this in our analysis, we construct an auxiliary graph containing such ``close'' vertices, which we use to group individual claw swaps into larger exchanges involving $O(k^2)$ claws. We then show that the total slack we gain across all such large exchange is a significant fraction of the weight of all the vertices of $A$ whose claws participate in the exchange. For the remaining vertices $a$ that have no ``close'' vertex $b$, we show that $\Psi_a$ is large. To gain some further intuition in this case, one can consider the example shown in Figure~\ref{fig:sqrt3}, which is the worst case when applying Berman's algorithm to a single, isolated claw. Here the locality gap is only $\sqrt{3} < \frac{3+1}{2}$, which we show can be attributed to the slack $\Psi_a$. We show that even when a claw is not strictly isolated, as long as all of the other vertices in its neighborhood have significantly smaller weight than that of its center vertex, there is still a relatively large amount of slack $\Psi_a$ that can be exploited.

By balancing these two cases, we are able to save over Berman's charging scheme for \emph{all} vertices in $A$, rather than only a subset of constant weight. This is enough to obtain a 1.811-approximation for weighted 3-set packing, which we present in Section~\ref{sec:larger-exchanges}. In Section~\ref{sec:bound-isolated-edges} we show that by roughly doubling the sizes of the swaps we consider it is possible to handle separately a key bottleneck case in our analysis and thus improve the ratio further to $1.786$.

Note that when $k = 3$, the small example in Figure \ref{fig:sqrt3} shows that we cannot attain an approximation factor smaller than $\sqrt{3} \approx 1.732 > 3/2$. Intuitively, our improvements increase because the gap between $\frac{k+1}{2}$ and the bound of $\sqrt{k}$ for an isolated single claw (as shown for $k = 3$ in Figure~\ref{fig:sqrt3}) grows larger as $k$ increases. 

\floatsetup{floatrowsep=qquad}
\begin{figure}
\begin{floatrow}
\capbtabbox[0.55\Xhsize]{%
\begin{tabular}{ccccc}\toprule
           Swap Size: & \multicolumn{2}{c}{$k^2(k-1)+k$} & \multicolumn{2}{c}{$2k^2(k-1) + k$} \\ \cmidrule(lr){1-1} \cmidrule(lr){2-3} \cmidrule(lr){4-5}
        $k$ & $\tau'_k/2$ & APX & $\tau_k/2$ & APX \\ \midrule
          3 & 0.189 & 1.811 & 0.214 & 1.786 \\
          4 & 0.210 & 2.290 & 0.251 & 2.249 \\
          5 & 0.219 & 2.781 & 0.269 & 2.731 \\
          6 & 0.225 & 3.275 & 0.281 & 3.219 \\
          7 & 0.229 & 3.771 & 0.289 & 3.711 \\
          8 & 0.232 & 4.268 & 0.294 & 4.206 \\
          9 & 0.234 & 4.766 & 0.299 & 4.701 \\
         10 & 0.236 & 5.264 & 0.302 & 5.198 \\ \bottomrule
        \end{tabular}
}{\caption{Approximation ratio for different values of $k$ and our  improvements $\tau'_k/2$, $\tau_k/2$, over $\frac{k+1}{2}$.}\label{tab:ratio}}
\ffigbox[\Xhsize]{%
      \begin{tikzpicture}[scale=1, every node/.style={scale=1}]
        \node[shape=circle, draw, thick, label=above:{$\sqrt{3}$}] (1) at (0,0) {};

        \node[shape=circle, thick, draw, label=below:{1}] (a) at (1,-2) {};
        \node[shape=circle, draw, thick, label=below:{1}] (b) at (0,-2) {};
        \node[shape=circle, draw, thick, label=below:{1}] (c) at (-1,-2) {};

        \path[thick] (a) edge node {} (1);
        \draw[thick] (b) -- (1);
        \draw[thick] (c) -- (1);

        \node at (-2,0) {$A:$};
        \node at (-2,-2) {$O:$};
      \end{tikzpicture}
      \vspace{0.5cm}
}{\caption{An isolated bad example for the weight-squared local search.}\label{fig:sqrt3}}
\end{floatrow}
\end{figure}

\subsection*{Further related work}
Nearly all algorithmic results for both the $k$-set packing problem and the maximum independent set problem in $(k+1)$-claw free graphs are based on variants of local search and greedy algorithms. In the unweighted setting, a simple local-search attempting to swap at most $2$ vertices into the current solution yields a $\frac{k+1+\e}{2}$-approximation. Hurkens and Schrijver \cite{Hurkens:1989:Size} showed that by considering swaps that add $O(\epsilon^{-1})$ vertices gives a $\frac{k+\e}{2}$-approximation. They also show that their analysis is tight, in the sense that any local-search which swaps a constant number of vertices has approximation factor at least $\frac{k+\e}{2}$ \cite{Hurkens:1989:Size}.
In contrast, Halld\'orson \cite{Halldorson:1995:Approximating} proved that a pure local search algorithm performing non-constant size swaps $\Omega(\log n)$ achieves a $\frac{k+2+\e}{3}$-approximation. This analysis was refined by Cygan et al. \cite{Cygan:2013:Sell} to obtain a ratio equal to $\frac{k+1+\e}{3}$. Due to the large swap sizes, the previous two results yield quasi-polynomial time algorithms. Sviridenko and Ward~\cite{Sviridenko:2013:Large} and Cygan \cite{Cygan:2013:Improved} designed polynomial-time local search algorithms with approximation factors of $\frac{k+2+\e}{3}$ and $\frac{k+1+\e}{3}$, respectively, by using techniques from fixed-parameter tractability. F\"urer and Yu \cite{Furer:2014:Approximating} gave a $\frac{k+1+\e}{3}$ approximation algorithm with improved dependence on $\e$ and also gave an instance with locality gap $\frac{k+1}{3}$ for any algorithm using swaps of size $O(n^{1/5})$. For all algorithms considering swaps of size $O(\log n)$ rely on the underlying structure specific to the $k$-set packing problem to find swaps in polynomial time, and thus do not generalize to the maximum independent set problem in $(k+1)$-claw free graphs. 

In the weighted setting, Arkin and Hassin showed that the standard weighted local-search algorithm performing swaps of size $O(\epsilon^{-1})$ yields only a $k - 1 + \e$ approximation~\cite{Arkin:1998:Local}. Chandra and Halld\'orson~\cite{Chandra:2001:Greedy} showed that the associated locality gap could be circumvented by combining a greedy algorithm followed by a local-improvement strategy that always selects the best improvement at each stage, yielding a $\frac{2(k+1)+\e}{3}$-approximation. As we have already noted, Berman~\cite{Berman:2000:d/2} obtained a $\frac{k+1}{2}$ approximation by considering a local search guided by the squared weights and swaps of size $k$. For smaller swaps of size 2, Berman and Krysta \cite{Berman:2003:Optimizing} showed that a local search guided by $w^\alpha$, for an appropriately chosen $1 < \alpha < 2$ has an approximation factor of $0.667k$, $0.651k$, and $0.646k$ for $k = 3$, $k=4$, and $k > 4$, respectively.

The $k$-set packing problem has also been studied via linear programming hierarchies. In this context, Chan and Lau \cite{Chan:2012:Linear} give an LP-rounding algorithm with approximation ratio $k-1+\frac{1}{k}$ for $k$-set packing and $k-1$ for $k$-dimensional matching. They also show that even after the linear program is strengthened by a linear number of rounds of the Sherali-Adams lifting procedure, its integrality gap remains at least $k-2$. In contrast, they show that by including a polynomial number of extra constraints, the integrality gap can be reduced to $\frac{k+1}{2}$. Singh and Talwar~\cite{Singh:2010:Improving} showed that the same integrality gap of $\frac{k+1}{2}$ can be achieved by applying $O(k^2)$ rounds of Chv\'atal-Gomory cuts to natural LP for the $k$-set packing problem.

\input{preliminaries}
\input{basic-analysis}
\input{isolated-edge}

\section{Our final guarantees}
In the previous sections, we have shown how to translate local optimality with respect $s$-exchanges into guarantees depending on a given parameter $0 \leq \epsilon \leq 1/2$. Here, we give a final, detailed version of Theorem~\ref{thm:main}, that gives concrete guarantees for various values of $k$ and also quantifies the asymptotic behavior of our guarantees in $k$. The exact numbers for the value of $\e$ and the improvement over the factor $\frac{k+1}{2}$ are displayed in Table \ref{tab:ratio-full}.

\begin{table}[h!]
\begin{tabular}{ccccccc}
\toprule
           Swap Size: & \multicolumn{3}{c}{$k^2(k-1)+1$} & \multicolumn{3}{c}{$2k^2(k-1) + 1$} \\ \cmidrule(lr){1-1} \cmidrule(lr){2-4} \cmidrule(lr){5-7}
        $k$ & $\tau_k/2$ & APX & $\e$ & $\tau_k/2$ & APX & $\e$ \\ \midrule
          3 & 0.189 & 1.811 & 0.3918 & 0.214 & 1.786 & 0.4533 \\
          4 & 0.210 & 2.290 & 0.2753 & 0.251 & 2.249 & 0.3281 \\
          5 & 0.219 & 2.781 & 0.2144 & 0.269 & 2.731 & 0.2614 \\
          6 & 0.225 & 3.275 & 0.1759 & 0.281 & 3.219 & 0.2176 \\
          7 & 0.229 & 3.771 & 0.1494 & 0.289 & 3.711 & 0.1866 \\
          8 & 0.232 & 4.268 & 0.1298 & 0.294 & 4.206 & 0.1635 \\
          9 & 0.234 & 4.766 & 0.1148 & 0.299 & 4.701 & 0.1455 \\
         10 & 0.236 & 5.264 & 0.1029 & 0.302 & 5.198 & 0.1311 \\ \bottomrule
        \end{tabular}
        \caption{Optimal settings for $\epsilon$ and approximation ratio for different values of $k$. Here, $\tau_k/2$, $\tau'_k$ measure the improvement over $\frac{k+1}{2}$.}
        \label{tab:ratio-full}
\end{table}
\def\thetheorem{\ref{thm:main}}
\begin{theorem}[Full]
Algorithm~\ref{alg:local-search-main} has approximation factor $\frac{k+1-\tau_k}{2}$, where
\begin{itemize}
\item For $s = k(k-1) + 1$, $\tau_k = \max_{\epsilon \in [0,1/2]}\min\left\{\frac{1-\epsilon}{2-\epsilon},\, \left(k - \frac{1}{\sqrt{1-\epsilon}}\right)(1 - \sqrt{1-\epsilon})\right\}$ for any $0 \leq \epsilon \leq 1/2$ and is non-decreasing in $k$ with $\lim_{k \to \infty}\tau_k = 1/2$.
\item For $s = 2k(k-1) + 1$, $\tau_k = \max_{\epsilon \in [0,1/2]}\min\left\{\frac{2(1-\epsilon)}{3-\epsilon},\, \left(k - \frac{1}{\sqrt{1-\epsilon}}\right)(1 - \sqrt{1-\epsilon}),\, 2 - \frac{1}{\sqrt{1-\epsilon}}\right\}$ for any $0 \leq \epsilon \leq 1/2$ and is non-decreasing in $k$ with $\lim_{k \to \infty}\tau_k = 2/3$. 
\end{itemize}
\end{theorem}

\begin{proof}[Proof of Theorem \ref{thm:main}]
The value for $\tau_k$ follows directly from Theorems~\ref{thm:bound-1} and \ref{thm:bound-2}. Note that for every $\epsilon \in [0,1/2]$, the second term in each maximum is an increasing function of $k$, and the remaining terms are constant. Thus, $\tau_k$ is a non-decreasing function of $k$ as claimed. 

For the rest of the proof we remove the subscript $k$ from $\tau_k$ and simply write $\tau$. In order to find the value for $\epsilon$ defining $\tau$ for a given $k$, one can equate the guarantees $\left(k-\frac{1}{\sqrt{1-\epsilon}}\right)(1-\sqrt{1-\epsilon})$ and $\frac{1-\epsilon}{2-\epsilon}$ or $\frac{2(1-\epsilon)}{3-\epsilon}$, respectively. This leads to a quartic equation in $\sqrt{1-\epsilon}$ whose solution is unwieldy. Thus, here we have instead given numerical solutions for $3 \leq k \leq 10$ in Table~\ref{tab:ratio-full} which are easily verified by substituting the given $\epsilon$ into each term of the minimum. 
For larger values of $k$, we employ a rough bound in order to show the claimed asymptotic behavior, rather than the precise rate of convergence of $\tau$ to $1/2$ and $2/3$, respectively. Specifically, we set $\epsilon = 1 - \left(\frac{k-1}{k}\right)^2$ and note that $\sqrt{1-\epsilon} = \frac{k-1}{k}$ and for all $k \geq 4$, $0 \leq \epsilon < \frac{1}{2}$. Then, for all $k \geq 4$,
\begin{gather*}
2-\frac{1}{\sqrt{1-\epsilon}} = 2 - \frac{k}{k-1}, \\
\left(k - \frac{1}{\sqrt{1-\epsilon}}\right)(1-\sqrt{1-\epsilon})
= \left(k - \frac{k}{k-1}\right)\left(1 - \frac{k-1}{k}\right) = k - (k-1) - \frac{k}{k-1} + 1 = 2-\frac{k}{k-1},
\\
\frac{1-\epsilon}{2-\epsilon} \leq \frac{1}{2},\\
\frac{2(1-\epsilon)}{3-\epsilon} \leq \frac{2}{3}.
\end{gather*}
Now, we note that $2 - \frac{k}{k-1} = 1 - \frac{1}{k-1} \geq \frac{2}{3}$ for all $k \geq 4$. Thus, for $k \geq 4$, we have $\tau = \frac{1-\epsilon}{2-\epsilon}$ for $s = k(k-1) + 1$, and $\tau = \frac{2(1-\epsilon)}{3-\epsilon}$ for $s = 2k(k-1) + 1$. It then suffices to note that $\lim_{k\to\infty} \epsilon = \lim_{k \to \infty} 1 - \left(\frac{k-1}{k}\right)^2 = 0$.
\end{proof}

\section{Conclusion and Future directions}
The central result of our paper is the design and analysis of a large neighborhood search algorithm for finding maximum weight independent set in $(k+1)$-claw free graphs. We prove that an $O(k^3)$ neighborhood search is sufficient to get  an approximation ratio of $1.786$ for weighted 3-set packing, and has an asymptotic approximation guarantee equal to $\frac{k+\frac{1}{3}}{2} $ as $k \to \infty$, substantially improving upon Neuwohner's bound of $\frac{k+1}{2} - \frac{1}{63,700,992}$ for $k = 3$ as well as subsequent improvements for moderately large\footnote{We remark briefly that for $k = 10$ we obtain an improvement of $0.302$ over $\frac{k+1}{2}$, which gives a lower bound on the improvements we obtain for all $k \geq 10$. For all $k \leq 229$, this lower bound is larger than the improvement of $\frac{k+1}{2} - (0.4986 (k + 1) + 0.0208)$ given in~\cite{DBLP:journals/corr/abs-2202-01248}.}. While we obtain an improved approximation ratio, our main objective is to provide tools to analyze Berman's algorithm with respect to large exchanges. We believe that our techniques are versatile and could be extended further to obtain a $\sqrt{3}$-approximation for $k=3$ and a $\frac{k}{2}$-approximation for $k \geq 4$. In particular, guided by Example \ref{fig:bad-ex}, it would be interesting to analyze the effect of large swap for isolated vertices, which leaves room for future improvements.

\bibliographystyle{plain}
\bibliography{set-pack-bib-minimal}
\appendix
\input{appendix}

\end{document}

%% file: preliminaries.tex
\section{Preliminaries}
\label{sec:preliminaries}
In this section, we fix the notation used throughout the remaining paper. We consider the general setting in which we are given a vertex-weighted $(k+1)$-claw free graph $G = (V,E)$ and seek an independent set maximum  weight. For each $v \in V$, we let $w_v \in \mathbb{R}_{+}$ denote the given weight of $v$ and for any $A \subseteq V$ we let $w(A)$ denote the total weight $\sum_{v \in A}w_v$ of all vertices in $A$.

For any two subsets $A, B$ of vertices in $V$  we define the \emph{neighbourhood of $A$ in $B$}, written $N(A, B)$, as $N(A,B) \triangleq \{b \in B \colon (a,b) \in E(G) \text{ for some } a\} \cup (A \cap B)$. 
To simplify notation, we will write $N(o, A)$ instead of $N(\{o \}, A)$ for a vertex $o \in V$, and additionally use the shorthand $A - a$ for $A \setminus \{a\}$.
Because $G$ is $(k+1)$-claw free $N(v,V)$ contains at most $k$ pairwise non-adjacent vertices for every $v \in V$. In particular, if $A$ is an independent set of vertices, then $|N(v,A)| \leq k$ for all $v \in V$ and  $N(v,A) = \{v\}$ for all $v \in A$.

The general local search procedure that we analyze is shown in Algorithm~\ref{alg:local-search-main}. The procedure maintains a current solution $S$, which we initialize using the standard greedy algorithm. We let $s \geq 1$ be a parameter governing the size of the exchanges performed by the algorithm. The algorithm repeatedly searches for an independent set of at most $sk$ vertices $C \subseteq V \setminus S$ with total \emph{squared} weights larger than the total squared weight of the conflicting vertices $N(C,S)$ in $S$. Whenever such a set is found, the algorithm adds $C$ to $S$ and removes $N(C,S)$ from $S$. Formally, for any $A \subseteq V$, we let $w^2(A) \triangleq \sum_{v \in A} w^2_v$. Then, Algorithm~\ref{alg:local-search-main} exchanges a set $C \subseteq V\setminus S$ for $N(C,S)$ only if $w^2(C) > w^2(N(C,S))$.  We can implement the search for each improvement in time $O(n^{sk})$ via simple enumeration. By using a pre-processing procedure to rescale and round the input weights, it can be ensured that the algorithm converges to an a local optimum in polynomial time while suffering a slight loss of approximation~\cite{Berman:2000:d/2}. In fact, because this results in only a small, polynomial dependence on this loss factor, a simple partial enumeration procedure can be used to remove the loss entirely, as we show in Appendix~\ref{sec:conv-algor}.

\begin{algorithm2e}[t]
\caption{Squared Weight Local Search with $s$-Exchanges}
 \label{alg:local-search-main}
\SetKw{Break}{break}
\SetArgSty{textrm}
$S \gets $ the output of the standard greedy algorithm applied to $G$ and $w$\;
\Repeat{$S = S'$}
{
  $S' \gets S$\;
  \ForEach{$C \subseteq V\setminus S$ of containing at most $sk$ vertices}{
    \If{$C$ is an independent set and $w^2(C) > w^2(N(C,S))$}
    {
      $S' \gets S \cup C \setminus N(C,S)$\;
      \Break\;
    }
  }
}
\Return{$S$}
\end{algorithm2e}


Thus, in all of our remaining analysis, we will suppose that the algorithm has terminated and produced a locally optimal solution $A$ for our instance. We let $O$ denote the optimal solution of this same instance. Note that both $A$ and $O$ are independent sets of $G$, and since $G$ contains no $(k+1)$-claw, the maximum degree in the subgraph of $G$ induced by $A \cup O$ is at most $k$. 

In order to define a set of claw swaps, Berman~\cite{Berman:2000:d/2} makes a mapping
$\pi\colon O \to A$ by $\pi(o) = \arg\max\{w_x : x \in N(o,A)\}$, breaking ties in an arbitrary, consistent manner. Note that $\pi(o)$ is the neighbour of $o$ in $A$ of maximum
weight. Using
$\pi$, we define a collection of sets
$\cC = \{\claw{a}\}_{a \in A}$, where $C_a \triangleq \{o : \pi(o) = a\}$.
Then, each vertex $o \in O$ appears in exactly one set $\claw{a} \in \cC$. We observe each set $\claw{a}$ forms the talons of a claw of $G$ centered at vertex $a \in A$. Thus $\card{\claw{a}} \leq k$ for all $a \in A$. 
Moreover for each $a \in A$, we have $w_a \geq w_v$ for all $v \in
N(\claw{a},A)$.

For each $a \in A$, we define $N^+_a \triangleq \{a\} \cup \bigcup_{o \in \claw{a}}N(o,A-a)$. Note that if $\claw{a} \neq \emptyset$ then $N_a^+ = N(\claw{a},A)$ and if $\claw{a} = \emptyset$ then $N_a^+ = \{a\}$. For each $a \in A$, we consider in our analysis a local operation that adds $\claw{a}$ to $A$ and removes $N^+_a$ from $A$. We call each such operation a \emph{$1$-exchange}, since it involves the talons of one claw $\claw{a}$. Local optimality with respect to these 1-exchanges then implies that for any $a \in A$,
\begin{equation}
  w^2(\claw{a}) \leq w^2(N^+_a) \leq w^2_a + \sum_{o \in \claw{a}}w^2(N(o,A-a)),
\label{eq:neighbourhood-lo}
\end{equation}
where the final inequality follows since $a \in N(o,A)$ for all $o \in \claw{a}$.  Note that for empty claws with $\claw{a} = \emptyset$, the above inequality follows immediately from $N^+_a = \{a\}$ and non-negativity of the weights $w_a$.
In this case, observe that the corresponding 1 exchange simply removes $a$ from the solution $A$.


\section{The $(k+1)/2$-approximation algorithm of Berman}
\label{sec:berman}

We now review the argument from the analysis of Berman~\cite{Berman:2000:d/2}, which shows that the absence of improving 1-exchanges for $w^2$ implies that $w(O) \leq \frac{k+1}{2}w(A)$-approximation. Berman's proof uses a $2$-stage charging argument and shows that each vertex in the current solution $A$ receives less than $(k+1)/2$ times its weight. Here we present a (arguably) simpler proof without charging argument, in which we make explicit the slack in several inequalities that are key in the analysis of~\cite{Berman:2000:d/2}.
For each $a \in A$, and $o \in \claw{a}$, we define the following quantities to measure this slack:
\begin{align*}
\psi_{a,o} &\triangleq (w_o - w_a)^2 + w_aw(N(o,A-a)) - w^2(N(o,A-a)),\\
\Psi_a &\triangleq \sum_{o \in \claw{a}} \psi_{a,o},\\
\Delta_a &\triangleq w^2(N^+_a) - w^2(\claw{a}).
\end{align*}
Consider first $\Psi_a$ and note that for each $a \in A$ and $o \in \claw{a}$, $(w_o - w_a)^2 \geq 0$ and by construction of $\claw{a}$, $w_v \leq w_a$ for all $v \in N(o,A)$. Thus, $w^2(N(o,A-a)) = \sum_{v \in N(o,A-a)}w^2_v \leq w_a\sum_{v \in N(o,A-a)}w_v = w_aw(N(o,A-a)$ and so $\psi_{a,o} \geq 0$ for all $o \in \claw{a}$. It then follows that $\Psi_a \geq 0$ for all $a \in A$. Next, note that since $|\claw{a}| \leq k$ for each $a$, local optimality with respect 1-exchanges~\eqref{eq:neighbourhood-lo} implies that $\Delta_a \geq 0$ for all $a \in A$. We now show that the values $\Psi_a$ and $\Delta_a$ can indeed be treated as slack in the analysis of Berman's algorithm:
\begin{lemma}
\label{lem:berman-main}
Suppose $A$ is locally optimal with respect to 1-exchanges. Then,
\begin{equation*}
2w(O) \leq w(A) + \sum_{o \in O}w(N(o,A)) - \sum_{a \in A}\left[\frac{\Delta_a}{w_a} + \frac{\Psi_a}{w_a}\right].
\end{equation*}
\end{lemma}
\begin{proof}
Fix a single claw $\claw{a}$ and $o \in \claw{a}$. Then,
\begin{align}
\label{eq:berman-single-o}
2w_ow_a &= w_o^2 + w_a^2 - (w_o - w_a)^2\\
&= w_o^2 + w_a^2 - (w_o - w_a)^2
- w^2(N(o,A-a))
+ w^2(N(o,A-a)) \notag
\\
&\qquad
- w_aw(N(o,A-a))
+ w_aw(N(o,A-a))
  \notag \\
&= w_o^2 + w_a^2 - w^2(N(o,A-a)) + w_aw(N(o,A-a)) - \psi_{a,o}.\notag
\end{align}
Equation \eqref{eq:berman-single-o} holds for every $o \in \claw{a}$. Summing over all $o \in \claw{a}$ then gives:
\begin{align*}
2w_aw(\claw{a}) &= |\claw{a}|w_a^2 + w^2(\claw{a}) - \sum_{o \in \claw{a}}w^2(N(o,A-a)) + w_a\sum_{o\in \claw{a}}w(N(o,A-a)) - \Psi_a \\
&\leq (|\claw{a}| + 1)w_a^2 + w^2(\claw{a}) - w^2(N^+_a) + w_a\sum_{o\in \claw{a}}w(N(o,A-a)) - \Psi_a \\
&= (|\claw{a}| + 1)w_a^2 - \Delta_a + w_a\sum_{o\in \claw{a}}w(N(o,A-a)) - \Psi_a
\\
&= w_a^2 - \Delta_a + w_a\sum_{o\in \claw{a}}w(N(o,A)) - \Psi_a,
\end{align*}
where the inequality follows from the second inequality in~\eqref{eq:neighbourhood-lo}, and the final equation from the fact that $a \in N(o,A)$ for all $o \in \claw{a}$ by construction, and so $w^2_a + w_aw(N(o,A-a)) = w_aw(N(o,A))$ for each $o \in \claw{a}$. Dividing both sides by $w_a$ gives
\begin{equation}
\label{eq:berman-claw}
2w(\claw{a}) \leq w_a + \sum_{o\in \claw{a}}w(N(o,A)) - \left[\frac{\Delta_a}{w_a} + \frac{\Psi_a}{w_a}\right],
\end{equation}
which holds for each $a \in A$. Summing~\eqref{eq:berman-claw} over all $a \in A$ and recalling that each $o \in O$ appears in exactly one set $\claw{a} \in \cC$ then completes the proof.
\end{proof}


As an immediate corollary, we recover the standard approximation result of Berman~\cite{Berman:2000:d/2}.
\begin{corollary}
\label{cor:(k+1)/2}
For any $A$ that is locally optimal with respect to 1-exchanges, $w(O) \geq \frac{k+1}{2}w(A)$.
\end{corollary}
\begin{proof}[Proof of Corollary \ref{cor:(k+1)/2}]
As we have noted above, we have $\Psi_a \geq 0$ for all $a \in A$ and since $A$ is locally optimal with respect to 1-exchanges, $\Delta_a \geq 0$ for all $a \in A$. Thus, Lemma~\ref{lem:berman-main} implies that
\begin{equation*}
2w(O) \leq w(A) + \sum_{o \in O}w(N(o,A)).
\end{equation*}
Now, we note that since $O$ is an independent set and $G$ is $(k+1)$-claw free, each $a \in A$ appears in $N(o,A)$ for at most $k$ distinct $o \in O$. Thus, $\sum_{o \in O}w(N(o,A)) \leq k w(A)$. Using this in the inequality above and dividing through by 2 then completes the proof.
\end{proof}


%% file: basic-analysis.tex
\section{An Improved Algorithm Using Larger Exchanges}
\label{sec:larger-exchanges}

We now show that when $A$ is locally optimal with respect to larger exchanges, we can obtain a better approximation ratio. Our proof will proceed by obtaining a lower bound on the total slack $\Delta_a$ and $\Psi_a$ for all vertices in Lemma~\ref{lem:berman-main}. Before proceeding we give some high-level intuition for our approach. From the proof of Corollary \ref{cor:(k+1)/2}, we see that the approximation ratio of Algorithm~\ref{alg:local-search-main} is close to $\frac{k+1}{2}$, only when both $\Psi_a$ and $\Delta_a$ are close to $0$. For a given vertex $a \in A$, $\Delta_a$ measures the \emph{tightness} of the $1$-exchanges centered at $a$, in the sense that having $\Delta_a$ equal to $0$ means that the $1$-exchanges centered at $a$ satisfies $w^2(C_a) = w^2(\Nap)$. Suppose there are two vertices $a, b$ such that $b \in \Napm$, and consider an exchange which attempts to add $C_a \cup C_b$ and removes $\Nap \cup N_b^+$. If this larger exchange is non-improving, we will show that we cannot have both $\Delta_a = 0$ and $\Delta_b = 0$. Intuitively, this follows since $b$ is counted once in $N_a^+ \cup N^+_b$ but once in \emph{both} $N^+_a$ and $N^+_b$. Assuming that $b$ has a large weight yields a substantial improvement.
  On the other hand, if all the vertices in $\Napm$ have low weight compared to $a$, then we show that we can bound the slack term $\Psi_a$ away from $0$ well.

Our general approach will consider $s$-exchanges bringing the talons of $s > 1$ claws into $A$ simultaneously. In order to define the set of $s$-exchanges we consider in the analysis, we make use of the following auxiliary graph.
\begin{Definition}[Exchange Graph $H_\epsilon$]
Fix $0 \leq \epsilon \leq 1$. Then we define the \emph{exchange graph} $H_\epsilon$ to be a directed graph with $V(H_\epsilon) = A$ and $E(H_\epsilon)$ containing an arc $(a,b)$ from $a$ to $b$, for each $b \neq a$ if and only if:
\begin{enumerate}
\item $a \in N^+_{b},$
\item $w_a \geq (1-\epsilon)w_b$,
\end{enumerate}
\label{def:swap-graph}
\end{Definition}
Note that for any arc $(a,b) \in E(H_\e)$, we have $(1-\epsilon)w_b \leq w_a \leq w_b$. In Figure~\ref{fig:graphs} we show an example of a graph $G$ and the corresponding exchange graph $H_\epsilon$. Note that the first condition of Definition~\ref{def:swap-graph} implies that in $H_\epsilon$ contains an arc $(x,y)$ or $(y,x)$ only if $x \in A$ and $y \in A$ are joined by a path of length 2 in $G[A \cup O]$. Since the maximum degree in $G[A \cup O]$ is $k$, there are at most $k(k-1)$ paths of length 2 ending at any vertex $x \in A$, and so the maximum degree of any vertex $x \in V(H_\epsilon)$ is $k(k-1)$. \\

We will refer to vertices of degree 0 in $H_\epsilon$ as \emph{isolated vertices} and let $I$ denote the set of all isolated vertices. We call the remaining vertices $D \triangleq A \setminus I$ \emph{non-isolated vertices}. We consider each type of vertex separately, and show that the total value of the slack term is large in both cases.
\begin{figure}
  \begin{subfigure}{\textwidth}
  \centering
    \begin{tikzpicture}[scale=0.95, every node/.style={scale=0.9}]
      \node[shape=circle, draw, thick, label={$a$}] (1) at (0,0) {};
      \node[shape=circle, draw, thick, label={$b$}] (2) at (3,0) {};
      \node[shape=circle, draw, thick, label={$c$}] (3) at (6,0) {};
      \node[shape=circle, draw, thick, label={$d$}] (4) at (9,0) {};
      \node[shape=circle, draw, thick, label={$e$}] (5) at (12,0) {};

      \node[shape=circle, thick, draw] (a) at (1,-2) {};
      \node[shape=circle, draw, thick] (b) at (0,-2) {};
      \node[shape=circle, draw, thick] (c) at (-1,-2) {};
      \node[shape=circle, draw, thick] (d) at (3,-2) {};
      \node[shape=circle, draw, thick] (e) at (4,-2) {};
      \node[shape=circle, draw, thick] (f) at (6,-2) {};
      \node[shape=circle, draw, thick] (g) at (7,-2) {};
      \node[shape=circle, draw, thick] (h) at (9,-2) {};
      \node[shape=circle, draw, thick] (i) at (10,-2) {};
      \node[shape=circle, draw, thick] (i) at (10,-2) {};
      \node[shape=circle, draw, thick] (j) at (12,-2) {};

      \path[->, thick] (a) edge node {} (1);
      \draw[->, thick] (b) -- (1);
      \draw[->, thick] (c) -- (1);
      \draw[thick] (a) -- (2);
      \draw[thick] (a) -- (3);

      \draw[->, thick] (d) -- (2) ;
      \draw[->, thick] (e) -- (2);
      \draw[thick] (e) -- (3);

      \draw[->, thick] (f) -- (3);
      \draw[->, thick] (g) -- (3);

      \draw[->, thick] (h) -- (4);
      \draw[thick] (g) -- (4);
      \draw[thick] (f) -- (4);
      \draw[->, thick] (i) -- (4);

      \draw[thick] (g) -- (5);
      \draw[thick] (i) -- (5);
      \draw[->,thick] (j) -- (5);

      \node at (-2,-0) {$A:$};
      \node at (-2,-2) {$O:$};

      \coordinate (NE) at (current bounding box.north east);
      \coordinate (SW) at (current bounding box.south west);

    \end{tikzpicture}
    \caption{Conflict graph}
    \label{fig:conflict-graph}
  \end{subfigure}

  \begin{subfigure}{\textwidth}
    \centering
    \begin{tikzpicture}[scale=0.95, every node/.style={scale=0.9}]
      \useasboundingbox ([yshift=0.5cm]NE) rectangle ([yshift=1.5cm]SW);
      \node[shape=circle, draw, thick, label={$a$}] (1) at (0,0) {};
      \node[shape=circle, draw, thick, label={$b$}] (2) at (3,0) {};
      \node[shape=circle, draw, thick, label={$c$}] (3) at (6,0) {};
      \node[shape=circle, draw, thick, label={$d$}] (4) at (9,0) {};
      \node[shape=circle, draw, thick, label={$e$}] (5) at (12,0) {};
      \node[shape=circle] (dummy) at (1,-2) {};

      \draw[->, thick] (2) to [bend left=25] (1);
      \draw[->, thick] (3) to [bend right=25] (1);
      \draw[->, thick] (3) to [bend left=25] (2);
      \draw[->, thick] (4) to [bend left=25] (3);
    \end{tikzpicture}
    \caption{Exchange graph $H_\e$}
    \label{fig:exch-graph}
  \end{subfigure}
  \caption{In this picture, we show the exchange graph $H_{1/4}$ (Figure \ref{fig:exch-graph}), coming from the conflict graph $G[A \cup O]$ in Figure \ref{fig:conflict-graph}. We assume that $w_a = w_b = w_c = 1$, $w_d = 4/5$, and $w_e = 1/2$. In Figure~\ref{fig:conflict-graph}, we label the edge from each vertex of $o$ to $\pi(o)$ with an arrow and assume that ties are broken by ordering vertices by label.}
  \label{fig:graphs}
\end{figure}
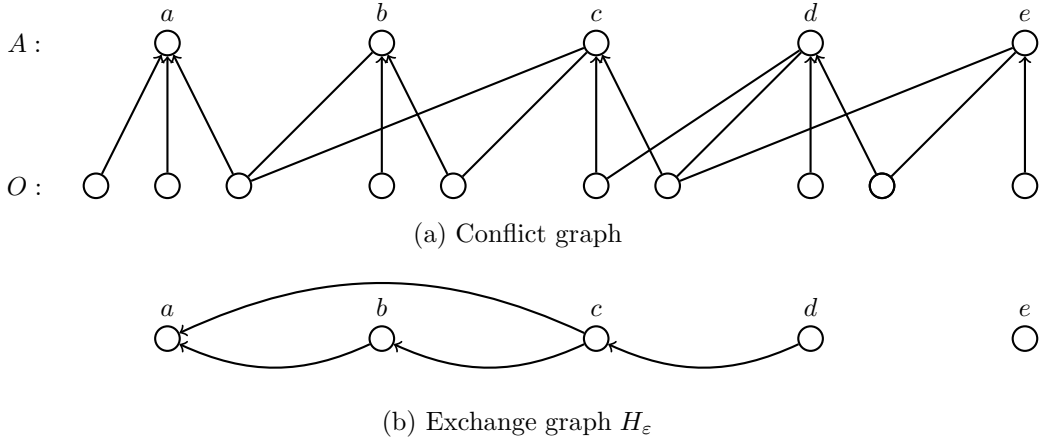

\subsection{Bounding the slack for non-isolated vertices}
\label{sec:bounding-slack-non}

In order to handle the non-isolated vertices $D$, we consider an \emph{$s$-exchange} in which $s > 1$ claws $\claw{a}$ are added together to $A$. We prove the following:

\begin{restatable}{lemma}{SetD}
  \label{lem:d-final}
  Let $0 \leq \epsilon \leq 1/2$ and suppose that $A$ is locally optimal under $s$-exchanges for $s \geq 1 + k(k-1)$. 
  Then,
  \begin{equation*}
  \sum_{a \in D}\left[\frac{\Delta_a}{w_a} + \frac{\Psi_a}{w_a}\right]
  \geq \frac{1-\epsilon}{2-\epsilon}w(D)+ \sum_{a \in D}\sum_{o \in \claw{a}}\epsilon w(N(o,I)).
  \end{equation*}
  \end{restatable}

\noindent Note that Lemma \ref{lem:d-final} implies that for vertices in $D$, we can save almost half of their total weight $(1-\e)/(2-\e)$. This contrasts with the proof of Berman's algorithm where we had only $\sum_{a \in A} \frac{\Delta_a}{w_a} \geq 0$.  The rest of the section is devoted to proving Lemma \ref{lem:d-final}.\\

\noindent\textbf{Constructing exchanges}: We start by constructing an appropriate set of $s$-exchanges for $s = k(k-1)$. To do this, we partition the vertices of $D$ as follows. Let $T$ initially be a collection of arcs from $H_\epsilon$ constituting an arbitrary undirected spanning tree in each connected component of $H_\epsilon$ (note that here we will ignore the direction of each arc). As long as $T$ contains an undirected path of length at least three, we remove one of the middle arcs (i.e.\,an arc incident on 2 vertices of degree 2) of this path from $T$. Observe that each such alteration decreases the number of arcs in $T$, and so this procedure terminates. At the end of the procedure, our final set of arcs $T$ is a collection of disjoint trees, each containing no path of length 3. This implies that each connected component of $T$ must be a star. Moreover, at the end of the procedure all vertices of $D$ have degree at least one in $T$, since we never remove an arc incident on a vertex of degree less than two. When the process terminates, it follows that $T$ is a disjoint collection of stars $T_1,\ldots,T_\ell$ contained in $H_\epsilon$, with each vertex of $D$ appearing in exactly 1 star. Since the maximum degree of a vertex in $H_\epsilon$ is at most $k(k-1)$, we have $|V(T_i)| \leq 1 + k(k-1)$ for all $i = 1,\ldots,\ell$ and $|\claw{a}| \leq k$ for each $a \in V(T_i)$. Thus, each such swap is an $s$-exchange adding an independent set of at most $sk$ vertices to $A$ and so will be considered by Algorithm~\ref{alg:local-search-main} when $s = k(k-1)$.\footnote{We briefly note that each claw except for the claw $\claw{v}$ associated with central vertex of a star $T_i$ shares an element of $O$ with $\claw{v}$. Thus one can in fact reduce the size of exchanges required by our algorithm from $k^2(k-1) + k$ to $k(k-1)^2 + k$. To avoid introducing further details, we have used a simpler bound throughout.} \\

\noindent Our first lemma shows that we can bound the total slack $\Delta_v$ for all $v \in T_i$ using the sum of the smallest weights $w_a$ associated with each arc $(a,b)$ in $T_i$.
\begin{lemma}
\label{lem:tree-edge-charge}
Let $\e \geq 0$ and suppose that $A$ is locally optimal under $s$-exchanges. Let $T$ be a tree in $H_\epsilon$ with $\card{V(T)} \leq s$, and $B \subseteq A$ be any set of vertices such that $H_\epsilon$ contains no arc $(u,v)$ or $(v,u)$ between any $u \in B$ and any $v \in V(T)$. Then,
\begin{align*}
    \sum_{v \in V(T)} \left[\frac{\Delta_v}{w_v} + \frac{\Psi_v}{w_v}\right] \geq \sum_{(a,b) \in E(T)}w_a + \sum_{v \in V(T)}\sum_{o \in \claw{v}}\epsilon w(N(o,B)).
\end{align*}
\end{lemma}
\begin{proof}
First we consider the slack terms $\Delta_v$. Recall that for each vertex $v$ we define $\Delta_v = w^2(N^+_v) - w^2(\claw{v})$. For any subset $X \subseteq A$, we similarly define $\Delta_X \triangleq w^2\left(\bigcup_{v \in X}N^+_v\right) - w^2\left(\bigcup_{v \in X}\claw{v}\right)$.
We first prove by induction on the size of $V(T)$ that:
\begin{equation}
\label{eq:swaps-induction-claim}
0 \leq \Delta_{V(T)} \leq \left(\sum_{v \in V(T)}\frac{\Delta_v}{w_v} - \sum_{(a,b) \in E(T)}\frac{
w^2_a}{w_b}\right)w_{\vmax},
\end{equation}
where $\vmax = \argmax_{v \in V(T)}w_v$. For the case in which $|V(T)| = 1$, we must have $V(T) = \{\vmax\}$, and $E(T) = \emptyset$. Thus,
\begin{equation*}
\left(\sum_{v \in V(T)}\frac{\Delta_v}{w_v} + \sum_{(a,b) \in E(T)}\frac{w^2_a}{w_b}\right)w_{\vmax} =
\frac{\Delta_{\vmax}}{w_{\vmax}} \cdot w_{\vmax} = \Delta_{\vmax} \geq 0,
\end{equation*}
as required, where the final inequality follows from local optimality with respect to 1-exchanges. Suppose now that \eqref{eq:swaps-induction-claim} holds for all trees $T$ with $|V(T)| \leq t < s$ and consider some tree $T$ with $|V(T)| = t+1$. As above, let $\vmax$ be a vertex of $V(T)$ with maximum weight and now let $T^1,\dots,T^c$ be the connected components of $T[V(T) - \vmax]$ obtained by removing $\vmax$. Then, each $T^i$ is a tree with $|V(T^i)| \leq t$ and the arcs incident to $\vmax$ in $T$ are of the form $(t_1,\vmax),\ldots,(t_c,\vmax)$, with $t_i \in V(T^i)$ for each $i = 1,\ldots,c$, (where the orientation of each arc follows from the fact that $w_{\vmax}$ is the largest weight in $V(T)$). Further let $\vmax_i = \arg\max_{v \in V(T_i)}w_v$. Then, local optimality with respect to $s$-exchanges implies that:
\begingroup
\allowdisplaybreaks
\begin{align*}
0 \leq \Delta_{V(T)} &= \textstyle w^2\left(\bigcup_{v \in V(T)}N^+_v\right) - w^2\left(\bigcup_{v \in V(T)}\claw{v}\right) \\
&\leq \sum_{i = 1}^c \textstyle \left[w^2\left(\bigcup_{v \in V(T_i)}N^+_v\right) - w^2\left(\bigcup_{v \in V(T_i)}\claw{v}\right)\right] + w^2(N^+_{\vmax}) - w^2(\claw{\vmax}) - \displaystyle \sum_{i = 1}^c w^2_{t_i}\\
&= \sum_{i = 1}^c \left[\Delta_{T_i}\right] + \Delta_{\vmax} - \sum_{i = 1}^c w^2_{t_i} \\
&\leq \sum_{i = 1}^c \left[
\left(\sum_{v \in V(T_i)}\frac{\Delta_v}{w_v} - \sum_{(a,b) \in E(T_i)}\frac{
w^2_a}{w_b}\right)w_{\vmax_i}\right] + \Delta_{\vmax} - \sum_{i = 1}^c w^2_{t_i} \\
&\leq \sum_{i = 1}^c \left[
\left(\sum_{v \in V(T_i)}\frac{\Delta_v}{w_v} - \sum_{(a,b) \in E(T_i)}\frac{
w^2_a}{w_b}\right)w_{\vmax}\right] + \frac{\Delta_{\vmax}}{w_{\vmax}}\cdot w_{\vmax} - \sum_{i = 1}^c \frac{w^2_{t_i}}{w_{\vmax}}\cdot w_{\vmax} \\
&= \left(\sum_{v \in V(T)}\frac{\Delta_v}{w_v} - \sum_{(a,b) \in E(T)}\frac{w^2_a}{w_b}\right)w_{\vmax}.
\end{align*}
\endgroup
Here, the second inequality follows from the fact that for each arc $e_i = (t_i,\vmax)$ between $T_i$ and $\vmax$, we have $t_i \in N^+_{\vmax}$ and $t_i \in N^+_{t_i}$. Thus, $w^2_{t_i}$ is counted in both $\bigcup_{v \in V(T_i)}w^2(N^+_v)$  and $w^2(N^+_{\vmax})$ but only once in $\bigcup_{v \in V(T)}N^+_v$. Moreover, each element of $O$ appears in at most 1 of the sets $\claw{v}$ and so $w^2\left(\bigcup_{v \in V(T)}\claw{v}\right) = w^2(\claw{\vmax}) + \sum_{i = 1}^cw^2\left(\bigcup_{v \in V(T_i)}\claw{v}\right)$.
The third inequality follows from the first inequality of the induction hypothesis~\eqref{eq:swaps-induction-claim}. The fourth inequality follows again from the second inequality of the induction hypothesis~\eqref{eq:swaps-induction-claim} and $w_{\vmax} \geq w_{v_i}$ for all $i$. This completes the induction step for the proof of~\eqref{eq:swaps-induction-claim}.

Rearranging~\eqref{eq:swaps-induction-claim}, for any $V(T)$ with $|V(T)| \leq s$ we have:
\begin{equation}
\sum_{v \in V(T)}\frac{\Delta_v}{w_v} \geq \sum_{(a,b) \in E(T)}\frac{w^2_a}{w_b}.
\label{eq:ni-delta-bound}
\end{equation}
Now we consider the slack $\Psi_v = \sum_{o \in \claw{v}}\psi_{v,o}$ for each vertex $v \in V(T)$. Recall that
\begin{multline*}
\Psi_{v} = \sum_{o \in \claw{v}}(w_o - w_v)^2 + w_vw(N(o,A-v)) - w^2(N(o,A-v))
\\ \geq \sum_{o \in \claw{v}}w_vw(N(o,A-v)) - w^2(N(o,A-v))
= \sum_{o \in \claw{v}}\sum_{a \in N(o,A-v)}\left[w_vw_a - w^2_a\right].
\end{multline*}
Each inner term in the final summation is non-negative, since for any $o \in \claw{v}$, $w_a \leq w_v$ for all $a \in N(o,A-v)$. For any $a \in N(o,B) \subseteq N(o,A-v)$ for some $o \in \claw{v}$, we must further have $w_a < (1-\epsilon)w_v$ since otherwise an arc $(a,v)$ would be present in $H_\epsilon$. Thus, for all $a \in N(o,B)$, $w_aw_v - w_a^2 \geq w_vw_a - (1-\epsilon)w_vw_a = \epsilon w_vw_a$. Finally, every $a$ with $(a,v) \in E(T)$ must appear in $N(o,(A\setminus B)-v)$ for some $o \in \claw{v}$. Thus,
\begin{align*}
\Psi_v \geq \sum_{o \in \claw{v}}\sum_{a \in N(o,A-v)}\left[w_v w_a - w^2_a\right] &= \sum_{o \in \claw{v}}\left[\sum_{a \in N(o,B)}\left[w_vw_a - w_a^2\right] + \sum_{a \in N(o, (A\setminus B) - v)}\left[w_vw_a - w_a^2\right]\right] \\
&\geq \sum_{o \in \claw{v}}\left[\sum_{a \in N(o,B)}\epsilon w_vw_a\right] + \sum_{a : (a,v) \in E(T)}\left[w_vw_a - w_a^2\right].
\end{align*}
Summing over all $v \in V(T)$ we obtain:
\begin{align}
\sum_{v \in V(T)} \frac{\Psi_v}{w_v} &\geq \sum_{v \in V(T)}\sum_{o \in \claw{v}}\sum_{a \in N(o,B)}\frac{\epsilon w_vw_a}{w_v} + \sum_{v \in V(T)}\sum_{a : (a,v) \in E(T)}\frac{w_vw_a - w_a^2}{w_v}  \label{eq:ni-psi-bound} \\
&= \sum_{v \in V(T)}\sum_{o \in \claw{v}}\epsilon w(N(o,B)) + \sum_{(a,b) \in E(T)}\left[w_a - \frac{w_a^2}{w_b}\right]. \notag
\end{align}
The claimed result then follows by combining the equations~\eqref{eq:ni-delta-bound} and \eqref{eq:ni-psi-bound}.
\end{proof}

In Lemma \ref{lem:tree-edge-charge}, the slack from the $s$-exchanges is expressed using the arcs of the tree $T$. Next, we show that this can in turn be bounded with respect to the total weight of all vertices of $T$.
\begin{lemma}
For any undirected tree $T$ contained in $H_\epsilon$, with $|V(T)| = t \geq 2$ and any $0 \leq \epsilon \leq 1/2$,
\begin{equation*}
\sum_{(a,b) \in E(T)}w_a \geq \frac{(t-1)(1-\epsilon)}{t - \epsilon}\sum_{v \in V(T)}w_v.
\end{equation*}
\label{lem:tree-bound}
\end{lemma}
\begin{proof}
Let $r$ be a vertex of $T$ with minimum weight, and fix some edge $(r,x) \in E(T)$. For each remaining arc $(a,b) \in E(T) - (r,x)$, consider the unique \emph{undirected} path from $r$ ending with $(a,b)$, and let $v \in \{a,b\}$ be the vertex at the end of this path. Note that every vertex of $V(T) \setminus \{r,x\}$  serves as $v$ for exactly one edge $(a,b) \in E(T) - (r,x)$. Moreover, if $v = a$, then $w_a = w_v$, and if $v = b$, then $w_a \geq (1-\epsilon)w_b = (1-\epsilon)w_v$. Let $z \triangleq \sum_{(a,b) \in E(T) - (r,x)}w_a$. Then, by the above discussion, $z \geq (1-\epsilon)\sum_{v \in V(T)\setminus \{r,x\}}w_v$. Consider the function:
\begin{equation*}
f(z) = \frac{w_r + z}{w_r + w_x + (1-\epsilon)^{-1}z} \leq \frac{\sum_{(a,b) \in E(T)}w_a}{\sum_{v \in V(T)}w_v}.
\end{equation*}
To complete the proof it suffices to show that $f(z) \geq \frac{(t-1)(1-\epsilon)}{t-\epsilon}$. Observe that since $(r,x) \in E(T)$, $(1-\epsilon)w_x \leq w_r \leq w_x$, and so $\frac{w_r}{w_r+w_x} \leq \frac{1}{2}$. Thus, for $\epsilon \leq 1/2$, $f(z)$ is a non-decreasing function of $z$. Moreover, since $w_r$ is the smallest weight of any vertex in $T$, $z \geq |E(T) - (r,x)|w_r = (t-2)w_r$. Thus,
\begin{equation*}
f(z) \geq \frac{w_r + (t - 2)w_r}{w_r + w_x + (1-\epsilon)^{-1}(t-2)w_r} \geq
\frac{(t-1)w_r}{w_r + (1-\epsilon)^{-1}w_r + (1-\epsilon)^{-1}(t - 2)w_r} = \frac{(t-1)(1-\epsilon)}{t - \epsilon},
\end{equation*}
where the second inequality follows again from $w_r \geq (1-\epsilon)w_x$.
\end{proof}

Combining Lemmas~\ref{lem:tree-edge-charge} and~\ref{lem:tree-bound} we then prove Lemma \ref{lem:d-final}.
\begin{proof}[Proof of Lemma \ref{lem:d-final}]
Recall that the collection of stars $T_1,\ldots,T_\ell$ has the property that each $a \in D$ appears in exactly one $T_i$ and $2 \leq |V(T_i)| \leq k(k-1) + 1 \leq s$. For each $T_i$, since $A$ is locally optimal with respect to $s$-exchanges, Lemma~\ref{lem:tree-edge-charge} (with $B = I$) and Lemma~\ref{lem:tree-bound}, respectively, imply
\begin{align*}
\sum_{a \in V(T_i)}\left[\frac{\Delta_a}{w_a}+\frac{\Psi_a}{w_a}\right] &\geq \sum_{(a,b)\in E(T_i)}w_a + \sum_{a \in V(T_i)}\sum_{o \in \claw{a}}\epsilon w(N(o,I))
\\ &\geq
\frac{1-\epsilon}{2-\epsilon}\sum_{a \in V(T_i)}w_a + \sum_{a \in V(T_i)}\sum_{o \in \claw{a}}\epsilon w(N(o,I)).
\end{align*}
Summing the resulting inequalities for each $i = 1,\ldots,\ell$ then gives the stated result. Note that the final inequality is tight only when $|V(T_i)| = 2$ for all $T_i$.
\end{proof}


\subsection{Bounding the slack for isolated claws}
\label{sec:bound-slack-isol}
We now consider those claws $\claw{a}$ where $a \in I$ is an isolated vertex. Observe that for all such claws, we must have $w_v \leq (1-\epsilon)w_a$ for every $v \in \bigcup_{o \in \claw{a}}N(o,A-a)$, since otherwise there would be an edge $(v,a)$ in $H_\epsilon$. It follows that for all $a \in I$, and $o \in \claw{a}$,
\begin{equation}
\label{eq:iso-neighborhood}
w^2(N(o,A-a)) \leq (1-\epsilon)w_aw(N(o,A-a)).
\end{equation}
In the following Lemma we derive a bound for the slack of isolated claws.
\begin{lemma}
\label{thm:isolated-hard}
Suppose that $A$ is locally optimal with respect to 1-exchanges. Let $0 \leq \delta \leq \epsilon$ and for $0 \leq t \leq k$, define $\rho_t \triangleq \frac{t(\epsilon - \delta)}{1-\delta} - \frac{\epsilon - \delta}{1-\epsilon}$. Then, for any $a \in I,$
\begin{equation*}
\frac{\Psi_a}{w_a} \geq \rho_{|\claw{a}|} w_a + \delta\sum_{o \in \claw{a}}w(N(o,A-a)).
 \end{equation*}
\end{lemma}
The crucial observation that Lemma \ref{thm:isolated-hard} captures is that the slack for isolated vertices is non-trivial. While this observation may seem obvious, we point out a subtle difficulty in handling it. Looking at bound in Lemma~\ref{lem:berman-main}, we see that for any $a \in A$, $w_a$ appears once in the bound for each $o \in N(a, O)$. In particular, the weight $w_a$ appears $\card{C_a} + \card{N(a, O \backslash C_a)} \leq k$ times. If $a$ is isolated, then \eqref{eq:iso-neighborhood} implies that our bound will be instead $(\card{C_a} + \card{N(a, O \backslash C_a)}(1-\e))w_a$. The issue is that when $\card{N(a, O \backslash C_a)}$ is small, we do not save enough to make any progress. To handle this situation, each vertex will keep a fraction of the slack equal to the parameter $\rho_{\card{C_a}}$ in Lemma~\ref{thm:isolated-hard}, leaving a smaller amount $(1-\delta)w_b$ for other isolated vertices $b \in \Napm$, where $\delta \leq \e$. The parameter $\delta$ can intuitively be thought of as a way of dividing the slack from $\Psi_a$ into one portion that pays for $a$ and another that will help pay for other isolated claws.
\begin{proof}[Proof of Lemma~\ref{thm:isolated-hard}]
Fix $a \in I$ and consider any $o \in \claw{a}$. Since $a \in I$, Equation~\eqref{eq:iso-neighborhood} implies that
\begin{equation*}
  \psi_{a,o} = (w_a - w_o)^2 + w_aw(N(o,A-a)) - w^2(N(o,A-a)) \geq (w_a - w_o)^2 + \epsilon w_a w(N(o,A-a))
\end{equation*}
for every $o \in \claw{a}$ and so
\begin{equation*}
  \Psi_a \geq \sum_{o \in \claw{a}}(w_a - w_o)^2 + \epsilon w_a \sum_{o \in \claw{a}}w(N(o,A-a)).
\end{equation*}
Define $\alpha_o \triangleq w_o/w_a$ and $\beta \triangleq \sum_{o \in \claw{a}}w(N(o,A-a))/w_a$. Then we can reformulate the above inequality as
\begin{align}
\Psi_a
  \geq \sum_{o \in \claw{a}}(w_a - \alpha_ow_a)^2 + \epsilon w_a\beta w_a 
  &= w_a^2\left(\sum_{o \in \claw{a}}(1- \alpha_o)^2 + \epsilon \beta\right) \label{eq:iso-psi-main} \\
  &= w_a^2\left(\sum_{o \in \claw{a}}(1-\alpha_o)^2 + (\epsilon - \delta)\beta\right) + w_a^2\delta\beta.\notag 
\end{align}
We now lower bound the bracketed expression on the right. Since $A$ is locally optimal with respect to 1-exchanges, \eqref{eq:neighbourhood-lo} and \eqref{eq:iso-neighborhood} imply that
\begin{equation*}
\sum_{o \in \claw{a}}w^2_o =  w^2(\claw{a}) \leq w^2(N^+_a) \leq w_a^2 + \sum_{\mathclap{o \in \claw{a}}}w^2(N(o,A-a)) \leq
  w_a^2+\sum_{\mathclap{o \in \claw{a}}}w_a(1-\epsilon)w(N(o,A-a)).
\end{equation*}
Reformulating this inequality in terms of the values $\alpha_o$ and $\beta$, gives us the following constraint:
\begin{equation*}
  \sum_{o \in \claw{a}}\alpha_o^2w_a^2 \leq w_a^2 + (1-\epsilon)w_a^2\beta.
\end{equation*}
Dividing through by $w_a^2$ and then rearranging, we obtain $\beta \geq \frac{\left(\sum_{o \in \claw{a}}\alpha_o^2\right) - 1}{1-\epsilon}$. Then, since $\epsilon - \delta \geq 0$,
\begin{align}
\sum_{o \in \claw{a}}( 1 - \alpha_o)^2 + (\epsilon - \delta)\beta
& \geq
\sum_{o \in \claw{a}}(1 - \alpha_o)^2 + (\epsilon - \delta)\frac{\left(\sum_{o \in \claw{a}}\alpha_o^2\right) - 1}{1-\epsilon} \label{eq:f-iso-bound-1}\\
&=
|\claw{a}| - 2\sum_{o \in \claw{a}}\alpha_o + \sum_{o \in \claw{a}}\alpha^2_o  + (\epsilon - \delta)\frac{\left(\sum_{o \in \claw{a}}\alpha_o^2\right) - 1}{1-\epsilon} \notag \\
&= |\claw{a}| - 2\sum_{o \in \claw{a}}\alpha_o  + \frac{1-\delta}{1-\epsilon}\sum_{o \in \claw{a}} \alpha_o^2 - \frac{\epsilon-\delta}{1-\epsilon} \notag \\
&\geq |\claw{a}| - 2\sum_{o \in \claw{a}}\alpha_o + \frac{1-\delta}{1-\epsilon}\cdot \frac{1}{|\claw{a}|}\left(\sum_{o \in \claw{a}} \alpha_o\right)^2  - \frac{\epsilon-\delta}{1-\epsilon},\notag 
\end{align}
where the last inequality follows from Cauchy-Schwarz.  We can express the above lower bound as $f(x)$ where $f(x) = \frac{1-\delta}{|\claw{a}|(1-\epsilon)}x^2 - 2x + |\claw{a}| - \frac{\epsilon - \delta}{1-\epsilon}$ and $x = \sum_{o \in \claw{a}}\alpha_o$. Then, note that $\frac{d^2f}{dx^2} = \frac{2(1-\delta)}{|\claw{a}|(1-\epsilon)} > 0$, so $f$ is convex in $x$, and when $x = \frac{|\claw{a}|(1-\epsilon)}{1-\delta}$, we have
$\frac{df}{dx} = \frac{2(1-\delta)x}{|\claw{a}|(1-\epsilon)} - 2 = 0$. Thus,
\begin{multline}
|\claw{a}| - 2\sum_{o \in \claw{a}}\alpha_o + \frac{1-\delta}{|\claw{a}|(1-\epsilon)}\left(\sum_{o \in \claw{a}} \alpha_o\right)^2  - \frac{\epsilon-\delta}{1-\epsilon}
\\ = f\left(\sum_{o \in \claw{a}}\alpha_o\right)  \geq f\left(\frac{|\claw{a}|(1-\epsilon)}{1-\delta}\right)= \frac{|\claw{a}|(\epsilon - \delta)}{1-\delta} - \frac{\epsilon-\delta}{1-\epsilon} = \rho_{|\claw{a}|}.
\label{eq:iso-f}
\end{multline}
Combining the inequalities \eqref{eq:iso-psi-main}, \eqref{eq:f-iso-bound-1}, and \eqref{eq:iso-f} we finally have
\begin{equation*}
\frac{\Psi_a}{w_a} \geq \frac{\rho_{|\claw{a}|}w_a^2  + \delta\beta w_a^2}{w_a}
= \rho_{|\claw{a}|}w_a  + \delta \sum_{o \in \claw{a}}w(N(o,A-a)),
\end{equation*}
as required. 
\end{proof}

We now combine the previous results into the following bound for the total slack associated with all vertices $a \in I$.

\begin{lemma}
\label{lem-i-final}
Suppose that $A$ is locally optimal with respect to 1-exchanges.
Let $0 \leq \delta \leq \epsilon$, let $B$ be any set of vertices in $H_\epsilon$, and for $0 \leq t \leq k$, define $\rho_t \triangleq \frac{t(\epsilon - \delta)}{1-\delta} - \frac{\epsilon - \delta}{1-\epsilon}$. Then,
\begin{equation*}
\sum_{a \in I}\left[\frac{\Delta_a}{w_a} + \frac{\Psi_a}{w_a}\right] \geq \sum_{a \in I}\left(\rho_{|\claw{a}|} - \delta|\claw{a}|\right)w_a +  \sum_{a \in I}\sum_{o \in \claw{a}}\delta w(N(o,B)).
\end{equation*}
\end{lemma}
\begin{proof}
For each $a \in I$, Lemma~\ref{thm:isolated-hard} gives
\begin{multline}
\frac{\Delta_a}{w_a} + \frac{\Psi_a}{w_a} \geq \rho_{|\claw{a}|}w_a + \sum_{o \in \claw{a}}\delta w(N(o,A-a))
\\ = \rho_{|\claw{a}|}w_a - \delta |\claw{a}|w_a  + \sum_{o \in \claw{a}}\delta w(N(o,A))
\geq \rho_{|\claw{a}|}w_a - \delta |\claw{a}|w_a  + \sum_{o \in \claw{a}}\delta w(N(o,B)).\label{eq:i-final-hard}
\end{multline}
The lemma then follows by summing~\eqref{eq:i-final-hard} over all $a \in I$.
\end{proof}

\subsection{Combining the Bounds}
\label{sec:improved-bound}
We now combine the bounds on the slack for isolated and non-isolated claws given by Lemmas~\ref{lem:tree-bound} and \ref{lem-i-final} with Lemma~\ref{lem:berman-main} to obtain a guarantee for Algorithm~\ref{alg:local-search-main}. In the following, we will set the parameter $\delta$ according to $\epsilon$ as $\delta \triangleq 1 - \sqrt{1-\epsilon}$. Then, straightforward algebraic manipulations (given in Appendix~\ref{sec:omitted-proofs}) lead to the following:
\begin{restatable}{lemma}{rhoDelta}
\label{lem:rhoDelta}
Let $0 \leq \epsilon < 1$ and set $\delta = 1 - \sqrt{1-\epsilon}$. Then, $\delta \leq \epsilon$ and for all $0 \leq t \leq k$,
\begin{equation*}
\rho_t \triangleq \frac{t(\epsilon-\delta)}{1-\delta} - \frac{\epsilon - \delta}{1-\epsilon} = \left(t - \frac{1}{\sqrt{1-\epsilon}}\right)(1-\sqrt{1-\epsilon})
= \left(t - \frac{1}{\sqrt{1-\epsilon}}\right)\delta.
\end{equation*}
In particular, for all $0 \leq t \leq k$, $\rho_t = \rho_k - (k - t)\delta$.
\end{restatable}
Having set $\delta = 1 - \sqrt{1-\epsilon}$, we can now state our main result for this section. Theorem \ref{thm:bound-1} proves that the approximation factor is a tradeoff between isolated and non-isolated claws.
\begin{theorem}\label{thm:bound-1}
For any $0 \leq \e \leq 1/2$, if $A$ is locally optimal under $s$-exchanges for $s \geq 1 + k(k-1)$ then,
\begin{equation*}
w(O) \leq \left[\tfrac{k + 1}{2} - \tfrac{1}{2}\min\left\{\tfrac{1-\epsilon}{2-\epsilon}, \rho_k\right\}\right]w(A),
\end{equation*}
where $\rho_k = (k - 1/\sqrt{1-\epsilon})(1 - \sqrt{1 - \epsilon})$.
\end{theorem}
\begin{proof}
Set $\delta = 1 - \sqrt{1-\epsilon}$ and observe that $\delta \leq \e$. Then, by Lemma~\ref{lem:d-final},
\begin{equation*}
\sum_{a \in D}\left[\frac{\Delta_a}{w_a} + \frac{\Psi_a}{w_a}\right]
\geq \frac{1-\epsilon}{2-\epsilon}w(D) + \sum_{a \in D}\sum_{o \in \claw{a}}\epsilon w(N(o,I)) \geq
\frac{1-\epsilon}{2-\epsilon}w(D) + \sum_{a \in D}\sum_{o \in \claw{a}}\delta w(N(o,I)),
\end{equation*}
and by Lemma~\ref{lem-i-final} (with $B = I$),
\begin{equation*}
\sum_{a \in I}\left[\frac{\Delta_a}{w_a} + \frac{\Psi_a}{w_a}\right] \geq \sum_{a \in I}\left(\rho_{|\claw{a}|} - \delta|\claw{a}|\right)w_a +  \sum_{a \in I}\sum_{o \in \claw{a}}\delta w(N(o,I)).
\end{equation*}
Recall that, by Lemma \ref{lem:rhoDelta}, when $\delta = 1 - \sqrt{1-\epsilon}$, $\rho_t = \rho_k - (k - t)\delta$ for $0 \leq t \leq k$. Thus, $\rho_{|\claw{a}|} - \delta|\claw{a}| = \rho_k - \delta(k - |\claw{a}|)  -  \delta|\claw{a}| = \rho_k - k\delta$. Combining this with the 2 bounds above and recalling that every $o \in O$ appears in $\claw{a}$, for exactly one $a \in A$, then gives
\begin{align*}
\sum_{a \in A}\left[\frac{\Delta_a}{w_a} + \frac{\Psi_a}{w_a}\right]
&\geq \frac{1-\epsilon}{2-\epsilon}w(D) + \sum_{a \in I}(\rho_{|\claw{a}|} - |\claw{a}|\delta)w_a + \sum_{a \in A}\sum_{o \in C_a}\delta w(N(o,I)) \\
&= \frac{1-\epsilon}{2-\epsilon}w(D) + (\rho_k - k\delta)w(I) + \sum_{o \in O}\delta w(N(o,I)).
\end{align*}
Combining this bound with Lemma~\ref{lem:berman-main} we then have:
\begin{align*}
2w(O) &\leq w(A) + \sum_{o \in O}w(N(o,A)) - \frac{1-\epsilon}{2-\epsilon}w(D) - (\rho_k - k\delta)w(I) - \sum_{o \in O}\delta w(N(o,I)) \\
&= w(A) + \sum_{o \in O}w(N(o,A \setminus I)) + (1-\delta)\sum_{o \in O}w(N(o,I)) - \frac{1-\epsilon}{2-\epsilon}w(D) - (\rho_k - k\delta)w(I) \\
&\leq w(A) + k w(A\setminus I)  + k(1-\delta)w(I) - \frac{1-\epsilon}{2-\epsilon}w(D) - (\rho_k - k\delta)w(I) \\
&= w(A) + k w(A) - \frac{1-\epsilon}{2-\epsilon}w(D) - \rho_k w(I) 
\leq (k+1)w(A) - \min\left\{\frac{1-\epsilon}{2-\epsilon}, \rho_k\right\} w(A),
\end{align*}
as required.
\end{proof}

In Section~\ref{sec:bound-isolated-edges} we describe how to balance the two upper bounds on our improvement given by Theorem~\ref{thm:bound-1} to obtain a 1.811 approximation for $k=3$. In Appendix~\ref{sec:matching-lower-bound} we provide an example showing that our analysis of isolated claws is nearly tight in this case, and further improvements require must require considering swaps that combine isolated claws or improving on the factor $\frac{1-\e}{2-\e}$ that we obtain for non-isolated claws. In the next section, we adopt the latter approach.


%% file: isolated-edge.tex
\section{Further improving the bound}
\label{sec:bound-isolated-edges}

The main bottleneck in our analysis of non-isolated claws occurs when $|V(T_i)| = 2$ for all $T_i$. In this case, each star is simply an isolated edge. Lemma~\ref{lem:tree-bound} implies that if we could ensure that $|V(T_i)| \geq 3$ for all $T_i$, then we could improve the gain of $\frac{1-\e}{2-\e}$ we obtain for non-isolated claws to $\frac{2(1-\e)}{3-\e}$. However, this is not possible when $H_\e$ contains maximal connected components of size 2. Formally, we say that a pair of vertices $a,b \in A$ is an isolated edge in $H_{\epsilon}$ if $(a,b) \in E(H_{\epsilon})$ and $a$ and $b$ are incident to no edges of $H_{\epsilon}$ other than $(a,b)$. In this section, we show that we can combine the basic techniques from the previous section to show that the total slack obtained for each such edge is already relatively large. We then show that by using slightly larger swaps we can ensure that the remaining non-isolated vertices of $H_\e$ can be partitioned into trees $T_i$ of size at least 3.

Let $I_2$ be the set of all isolated edges $(a,b)$ in the $H_{\epsilon}$. In this section, we let $I_1 \subseteq A$ denote the set of isolated vertices, $I_2 \subseteq A \setminus I_1$ be the union of all vertices appearing in an isolated edge and $D = A \setminus (I_1 \cup I_2)$ denote the set of vertices that are not isolated or part of an isolated edge. Note that each vertex of $I_2$ appears in exactly one isolated edge $(a,b)$.
The following Lemma bounds the total slack that we obtain from an isolated edge.
\begin{lemma}\label{lem:iso-edge-main-bound}
Suppose that $A$ is locally optimal with respect to $2$-exchanges and let $0 \leq \delta \leq \epsilon \leq 1/2$. For $0 \leq t \leq k$, define $\rho_t \triangleq \frac{t(\epsilon - \delta)}{1-\delta} - \frac{\epsilon - \delta}{1-\epsilon}$. Then for any isolated edge  $(a,b)$ be in $H_\epsilon$,
\begin{equation*}
\sum_{v \in \{a,b\}}\left[\frac{\Delta_v}{w_v} + \frac{\Psi_v}{w_v}\right] \geq \left(\left(1 - \frac{\epsilon - \delta}{1-\epsilon}\right) - k\delta\right) w_a + (\rho_{|\claw{b}|} - |\claw{b}|\delta)w_b + \sum_{v \in \{a,b\}}\sum_{o \in \claw{v}}\delta w(N(v,A)).
\end{equation*}
\end{lemma}
\begin{proof}[Proof of Lemma \ref{lem:iso-edge-main-bound}]
Consider an isolated edge $(a,b)$, and note that $a \in N^+_b$ and $w_a \leq w_b \leq (1-\epsilon)^{-1}w_a$. For all $v \in \bigcup_{o \in \claw{a}}N(o,A-a)$ we must have $w_v < (1-\epsilon)w_a$, since otherwise an additional arc $(v,a)$ would be present in $H_\epsilon$. Using that $\e \geq \delta$, it follows that
\begin{multline}\label{eq:slack-a}
\frac{\Psi_a}{w_a} \geq \frac{1}{w_a}\sum_{o \in \claw{a}}\left[w_aw(N(o,A-a)) - w^2(N(o,A-a))\right]
\geq \sum_{o \in \claw{a}}\epsilon w(N(o,A-a)), \\
\geq \sum_{o \in \claw{a}}\delta w(N(o,A - a))
= \sum_{o \in \claw{a}}\delta w(N(o,A)) - |\claw{a}|\delta w_a.
\end{multline}
Now, we consider the vertex $b$. Since $A$ is locally optimal with respect to $2$-exchanges and $a \in N^+_a \cap N^+_b$,
\begin{equation*}
w^2(\claw{a}) + w^2(\claw{b}) \leq w^2(N^+_a \cup N^+_b) \leq w^2(N^+_b) + w^2(N^+_b) - w^2_a,
\end{equation*}
and so
\begin{equation}
\Delta_a+\Delta_b
= w^2(N^+_a) - w^2(\claw{a}) + w^2(N^+_b) - w^2(\claw{b}) \geq w^2_a.
\label{eq:iso-edge-delta-bound}
\end{equation}
Note that for any $o \in \claw{b}$, we have $w_v < (1-\epsilon)w_b$ for all
$v \in N(o,A \setminus \{a,b\})$,
since otherwise there would be an additional arc $(v,b)$ in $H$ for each such $v$. For each $o \in \claw{b}$, we define $\alpha_o \triangleq w_o/w_b$ and further define $\beta \triangleq \sum_{o \in \claw{b}}w(N(o,A \setminus \{a,b\}))/w_b$ and $\gamma \triangleq w_a/w_b$. Then,
\begin{align}
\Delta_a + \Delta_b + \Psi_b &\geq w_a^2 +  \sum_{o \in \claw{b}}\left[(w_b - w_o)^2 + w_bw(N(o,A-b)) - w^2(N(o,A-b))\right]  \label{eq:iso-edge-slack-1}\\
&\geq w_a^2 + \sum_{o \in \claw{b}}\left[(w_b - w_o)^2 + \epsilon w_b w(N(o,A\setminus\{b,a\}))\right] + (w_b - w_a)w_a \notag \\
&= w_b^2\gamma^2 + w_b^2\sum_{o \in \claw{b}}\left[(1 - \alpha_o)^2\right] + w_b^2 \epsilon \beta + w_b^2(1 - \gamma)\gamma \notag \\
&= w_b^2\left(\sum_{o \in \claw{b}}\left[(1-\alpha_o)^2\right] + \epsilon \beta + \gamma\right)\notag \\
&= w_b^2\left(\sum_{o \in \claw{b}}\left[(1-\alpha_o)^2\right] + (\epsilon -\delta)\beta + \gamma \right) + w_b^2\delta\beta,\notag 
\end{align}
where the first inequality follows from~\eqref{eq:iso-edge-delta-bound}.
Since the 1-exchange bringing in $\claw{b}$ is non-improving, we must also have:
\begin{equation*}
\sum_{o \in \claw{b}}w_o^2 \leq w_b^2 + w_a^2 + w^2(N(\claw{b}, A\setminus \{a,b\}))
\leq w_b^2 + w_a^2 + (1-\epsilon)w_b\sum_{o \in \claw{b}}w(N(o, A\setminus \{a,b\})),
\end{equation*}
or, equivalently,
\begin{equation*}
w_b^2\sum_{o \in \claw{b}}\alpha_o^2 \leq w_b^2\left(1 + \gamma^2 + (1-\epsilon)\beta\right).
\end{equation*}
Rearranging, we obtain $\beta \geq \frac{1}{1-\epsilon}\left(\sum_{o \in \claw{b}}\alpha_o^2 - 1 - \gamma^2\right)$. Thus,
\begin{align}
\label{eq:iso-edge-f-1}\sum_{o \in \claw{b}}(1-\alpha_o)^2 + (\epsilon -\delta)\beta + \gamma
&\geq
\sum_{o \in \claw{b}}(1-\alpha_o)^2 + \frac{\epsilon-\delta}{1-\epsilon}\left(\sum_{o \in \claw{b}}\alpha_o^2 - 1 - \gamma^2\right) + \gamma \\
&= |\claw{b}| - 2\sum_{o \in \claw{b}}\alpha_o +
\frac{1-\delta}{1-\epsilon}\sum_{o \in \claw{b}}\alpha_o^2 - \frac{\epsilon-\delta}{1-\epsilon} - \frac{\epsilon-\delta}{1-\epsilon}\gamma^2 + \gamma \notag \\
&\geq \rho_{|\claw{b}|} - \frac{\epsilon-\delta}{1-\epsilon}\gamma^2 + \gamma \notag \\
&\geq \rho_{|\claw{b}|} + \left(1 - \frac{\epsilon-\delta}{1-\epsilon}\right)\gamma,\notag
\end{align}
where the penultimate inequality follows exactly as~\eqref{eq:iso-f} in the proof of Lemma~\ref{thm:isolated-hard} and the final inequality from $\gamma = \frac{w_a}{w_b} \leq 1$. Combining~\eqref{eq:iso-edge-slack-1} and \eqref{eq:iso-edge-f-1} we then have
\begin{multline}
\Delta_a + \Delta_b + \Psi_b
\geq \rho_{|\claw{b}|}w^2_b + \left(1 - \frac{\epsilon-\delta}{1-\epsilon}\right)\gamma w^2_b + \delta\beta w_b^2
\\ = \rho_{|\claw{b}|}w^2_b + \left(1 - \frac{\epsilon-\delta}{1-\epsilon}\right)w_bw_a  + \sum_{o \in \claw{b}}\delta w_b w(N(o,A\setminus \{a,b\})).
\label{eq:iso-edge-total-b}
\end{multline}
Finally, combining~\eqref{eq:slack-a} and \eqref{eq:iso-edge-total-b}, and using that $w_b \geq w_a$, gives
\begin{align*}
\frac{\Delta_a}{w_a} + \frac{\Psi_a}{w_a} &+
\frac{\Delta_b}{w_b} + \frac{\Psi_b}{w_b}
\geq
\frac{\Psi_a}{w_a} +
\frac{\Delta_a + \Delta_b + \Psi_b}{w_b}\\
&\geq \sum_{o \in \claw{a}}\delta w(N(o,A)) - |\claw{a}|\delta w_a +
\rho_{|\claw{b}|}w_b + \left(1 - \frac{\epsilon - \delta}{1 - \epsilon}\right) w_a + \sum_{o \in \claw{b}}\delta w(N(o,A\setminus \{a,b\}) \\
&=
\sum_{o \in \claw{a}}\delta w(N(o,A)) - |\claw{a}|\delta w_a +
\rho_{|\claw{b}|}w_b + \left(1 - \frac{\epsilon - \delta}{1 - \epsilon}\right)w_a \\
&\qquad+ \sum_{o \in \claw{b}}\delta w(N(o,A)) - |\claw{b}|\delta w_b - |N(a,\claw{b})|\delta w_a
  \\
&\geq \delta\sum_{o \in \claw{a}} w(N(o,A)) + \delta\sum_{o \in \claw{b}} w(N(o,A)) + \left(\left(1 - \tfrac{\epsilon - \delta}{1 - \epsilon}\right) - k\delta\right) w_a - \left(\rho_{|\claw{b}|} - |\claw{b}|\delta\right) w_b,\notag
\end{align*}
where in the final inequality we have used that $|\claw{a}| + |N(a,\claw{b})| \leq k$, since $G$ is $(k+1)$-claw free and $\claw{a} \cap \claw{b} = \emptyset$.
\end{proof}

\begin{lemma}
  \label{lem:tree-size-2}
Consider the connected components of $H_\epsilon$. For each connected component that is not an isolated vertex or an isolated edge, there is a collection of subgraphs $T_1,\ldots,T_p$ of this component, where each $T_i$ is a tree with $|E(T_i)| \geq 2$, $|V(T_i)| \leq 2k(k-1) + 1$, and each vertex of the connected component appears in exactly one subtree $T_i$.
\end{lemma}
\begin{proof}[Proof of Lemma \ref{lem:tree-size-2}]
Throughout the proof, we forget about the orientation of the arcs and perform the decomposition by treating $H_\e$ as an undirected graph.
Consider a connected component of $H_\epsilon$ with at least 3 vertices, and let $T$ be an arbitrary (undirected) spanning tree of this component. We repeatedly modify $T$ as follows: if there is some edge $(a,b)$ in $T$ so that there are at least 2 edges in every component of $T - (a,b)$, then we remove $(a,b)$ from $T$. We repeat this procedure until no such edge can be found. Observe that each alteration decreases the number of edges in $T$ and so this procedure must terminate. At the end of the procedure, we have some collection $T_1,\ldots,T_p$ of connected components, each of which is a subtree of $T$ containing at least 2 edges, and each vertex of $V(T)$ appears in exactly one of the components $T_1,\ldots,T_p$. It remains to bound the size of each component. Consider any component $T_i$. Note if $T_i$ has 2 vertices $u$ and $v$ of degree at least 3, then after removing any edge on the path from $u$ to $v$ in $T_i$, we obtain two trees each of which must contain 2 edges, since 2 of the edges originally incident to $u$ and 2 of the edges originally incident to $v$ remain. Thus, $T_i$ has at most 1 vertex of degree greater or equal to 3.

Suppose first that $T_i$ has such a vertex, and let this vertex be $r$. If any vertex $v$ of $T_i$ is at distance 3 from $r$ in $T_i$, after removing the first edge $(r,u)$ on the path from $r$ to $v$, we obtain a tree with the last 2 edges of this path and a tree with $r$ together with at least 2 other edges (since $r$ originally had degree at least 3). Thus, if $T_i$ has a vertex of $r$ degree 3 or larger, then all other vertices must have degree 2 or 1 and be at distance at most $2$ from $r$. Since the degree of $H_\epsilon$ is at most $k(k-1)$, this implies that any such $T_i$ contains at most $2k(k-1) + 1$ vertices.

Now suppose that $T_i$ does not have any vertex of degree larger than 2. Then, $T_i$ must be a path. If this path has at least 5 edges, then one of the ``middle'' edges of the path can be removed to obtain 2 paths each containing 2 edges. Thus, in this case $T_i$ is in fact a path of length at most 4, and so contains at most $5$ vertices.
\end{proof}

\begin{lemma}
\label{lem:d-final-2}
Let $0 \leq \delta \leq \epsilon \leq 1/2$ and suppose that $A$ is locally optimal under $s$-exchanges for $s \geq 1 + 2k(k-1)$. 
Then,
\begin{equation*}
\sum_{a \in D}\left[\frac{\Delta_a}{w_a}+\frac{\Psi_a}{w_a}\right]  \geq \frac{2(1-\epsilon)}{3-\epsilon} w(D) + \sum_{a \in D}\sum_{o \in \claw{a}}\delta w(N(o,I_1\cup I_2)).
\end{equation*}
\end{lemma}
\begin{proof}
Applying Lemma~\ref{lem:tree-size-2} to each connected component of $H_\epsilon$ that has at least 2 vertices, we obtain a set of trees $T_1,\ldots,T_\ell$ with each vertex $a \in D$ appearing in $V(T_i)$ for exactly one $1 \leq i \leq \ell$ and and $V(T_1),\ldots,V(T_\ell)$, with $3 \leq |V(T_i)| \leq s$ for all $1 \leq i \leq \ell$. Then, applying Lemma~\ref{lem:tree-edge-charge} with $B = I_1 \cup I_2$ and Lemma~\ref{lem:tree-bound} to each $T_i$ gives
\begin{align*}
\sum_{a \in V(T_i)}\left[\frac{\Delta_a}{w_a}+\frac{\Psi_a}{w_a}\right] &\geq \sum_{(a,b)\in E(T_i)}w_a + \sum_{a \in V(T_i)}\sum_{o \in \claw{a}}\epsilon w(N(o,I_1 \cup I_2))  \\
&\geq
\frac{2(1-\epsilon)}{3-\epsilon}w(T_i) + \sum_{a \in V(T_i)}\sum_{o \in \claw{a}}\epsilon w(N(o,I_1 \cup I_2)) \\
 &\geq
\frac{2(1-\epsilon)}{3-\epsilon}w(T_i) + \sum_{a \in V(T_i)}\sum_{o \in \claw{a}}\delta w(N(o,I_1 \cup I_2)).
\end{align*}
Summing the resulting inequalities for each $i = 1,\ldots,\ell$ then gives the stated result.
\end{proof}
Similarly to Theorem \ref{thm:bound-1}, the next theorem shows that the approximation factor is a trade-off between isolated claws, isolated edges, and claws centered at vertices in larger components.
\begin{theorem}\label{thm:bound-2}
If $A$ is locally optimal under $s$-exchanges for $s \geq 1 + 2k(k-1)$ then for any $0 \leq \epsilon \leq 1/2$,
\begin{equation*}
w(O) \leq \left[\tfrac{k + 1}{2} -
\tfrac{1}{2}\min\left\{\tfrac{2(1-\epsilon)}{3-\epsilon},\, \rho_k,\, 2 - \tfrac{1}{\sqrt{1-\epsilon}}\right\}\right]w(A),
\end{equation*}
where $\rho_k \triangleq (k - 1/\sqrt{1-\epsilon})(1- \sqrt{1-\epsilon})$.
\end{theorem}
\begin{proof}
As in the proof of Theorem~\ref{thm:bound-1}, set $\delta = 1 - \sqrt{1-\epsilon}$ and note that $\delta \leq \e$. Summing the inequality from Lemma~\ref{lem:iso-edge-main-bound} over all isolated edges $(a,b)$, and noting that each $v \in I_2$ appears in exactly one isolated edge, either as $a$ or $b$, we have:
\begin{equation*}
\sum_{v \in I_2} \left[\frac{\Delta_v}{w_v} + \frac{\Psi_v}{w_v}\right] \geq \sum_{v \in I_2}\min\left\{\ierat - k\delta,\, \rho_{|\claw{v}|} - |\claw{v}|\delta\right\}w_v + \sum_{v \in I_2}\sum_{o \in \claw{v}}\delta w(N(o,A)).
\end{equation*}
For $\delta = 1 - \sqrt{1-\epsilon}$, we have $1-\frac{\epsilon - \delta}{1-\epsilon} = 1- \frac{\sqrt{1-\epsilon} - (1-\epsilon)}{1-\epsilon} = 2 - \frac{1}{\sqrt{1-\epsilon}}$. Furthermore, by Lemma \ref{lem:rhoDelta}, we have $\rho_{|\claw{b}|} - |\claw{b}|\delta = \rho_k - (k - |\claw{b}|)\delta - |\claw{b}|\delta = \rho_k - k\delta$ for each $b$. Thus, the above inequality implies that:
\begin{align}
\sum_{v \in I_2} \left[\frac{\Delta_v}{w_v} + \frac{\Psi_v}{w_v}\right] &\geq \min\left\{\ieratsimp - k\delta,\, \rho_k - k\delta\right\}w(I_2) + \sum_{v \in I_2}\sum_{o \in \claw{v}}\delta w(N(o,A)) 
\label{eq:improved-final-i2}
\\
&\geq \min\left\{\ieratsimp - k\delta,\, \rho_k - k\delta\right\}w(I_2) + \sum_{v \in I_2}\sum_{o \in \claw{v}}\delta w(N(o,I_1 \cup I_2)).\notag 
\end{align}
Furthermore, Lemma~\ref{lem-i-final} (with $B = I_1 \cup I_2$) and~\ref{lem:d-final-2} imply that:
\begin{align}
\sum_{a \in I_1}\left[\frac{\Delta_a}{w_a} + \frac{\Psi_a}{w_a}\right] &\geq \left(\rho_k - k\delta\right)w(I_1) +  \sum_{a \in I_1}\sum_{o \in \claw{a}}\delta w(N(o,I_1 \cup I_2)), \label{eq:improved-final-i1}\\
\sum_{a \in D}\left[\frac{\Delta_a}{w_a}+\frac{\Psi_a}{w_a}\right]  &\geq \frac{2(1-\epsilon)}{3-\epsilon} w(D) + \sum_{a \in D}\sum_{o \in \claw{a}}\delta w(N(o,I_1\cup I_2)).\label{eq:improved-final-d}
\end{align}
Adding~\eqref{eq:improved-final-i2}, \eqref{eq:improved-final-i1}, and \eqref{eq:improved-final-d} gives the following bound on the total slack for our instance:
\begin{equation*}
\sum_{a \in A}\left[\frac{\Delta_a}{w_a} + \frac{\Psi_a}{w_a}\right] \geq
\frac{2(1 - \epsilon)}{3-\epsilon}w(D)  + \left(\min\left\{\ieratsimp,\, \rho_k\right\} - k\delta\right)w(I_1 \cup I_2) + \sum_{o \in O}\delta w(N(o,I_1 \cup I_2)).
\end{equation*}
Using this bound in Lemma~\ref{lem:berman-main} we finally obtain:
\begin{align*}
2w(O) &\leq w(A) + \sum_{o \in O}w(N(o,A)) - \frac{2(1-\epsilon)}{3-\epsilon}w(D) \\\
&\qquad\qquad -\left(\min\left\{\ieratsimp,\,\rho_k\right\} - k\delta\right)w(I_1\cup I_2) - \sum_{o \in O}\delta w(N(o,I_1 \cup I_2)) \\
&= w(A) + \sum_{o \in O}w(N(o,A \setminus (I_1 \cup I_2)))
+ (1-\delta)\sum_{o \in O}w(N(o,I_1 \cup I_2))\\
&\qquad\qquad - \frac{2(1-\epsilon)}{3-\epsilon}w(D) - \left(\min\left\{\ieratsimp,\,\rho_k\right\} - k\delta\right)w(I_1 \cup I_2) 
\end{align*}
Using the fact that for any $B \subseteq A$, each $a \in B$ appears in $N(o,B)$ for at most $k$ distinct values of $o$, we get that: $\sum_{o \in O}w(N(o, A \setminus (I_1 \cup I_2))) \leq k w(A \setminus (I_1 \cup I_2))$ and $\sum_{o \in O}w(N(o,I_1 \cup I_2)) \leq k w(I_1 \cup I_2)$. Replacing the above bounds in the previous computation, we obtain the desired result
\begin{align*}
2w(O)&\leq w(A) + k w(A\setminus (I_1 \cup I_2))  + k(1-\delta)w(I_1 \cup I_2) - \frac{2(1-\epsilon)}{3-\epsilon}w(D) \\
&\qquad\qquad- \left(\min\left\{\ieratsimp,\,\rho_k\right\} - k\delta\right)w(I_1 \cup I_2) \\
&= w(A) + k w(A) - \frac{2(1-\epsilon)}{3-\epsilon}w(D) - \min\left\{\ieratsimp,\,\rho_k\right\} w(I_1 \cup I_2) \\
&\leq (k+1)w(A) - \min\left\{\frac{2(1-\epsilon)}{3-\epsilon},\,\rho_k,\,\ieratsimp \right\} w(A) \qedhere
\end{align*}

\end{proof}



%% file: appendix.tex
\section{Appendix}
\subsection{Omitted Proofs}
\label{sec:omitted-proofs}
\rhoDelta*
\begin{proof}
When $\delta = 1 - \sqrt{1-\epsilon}$, we have
\begin{align*}
\rho_t &= t\frac{\epsilon - \delta}{1-\delta} - \frac{\epsilon-\delta}{1-\epsilon} \\
&= t\frac{\sqrt{1-\epsilon} - (1 - \epsilon)}{\sqrt{1-\epsilon}}
- \frac{\sqrt{1-\epsilon} - (1-\epsilon)}{1-\epsilon} \\
&= t(1 - \sqrt{1-\epsilon}) - \frac{1 - \sqrt{1-\epsilon}}{\sqrt{1-\epsilon}} \\
&= \left(t - \frac{1}{\sqrt{1-\epsilon}}\right)\delta.
\end{align*}
Further note that for all $0 \leq \epsilon < 1$, $1- \sqrt{1-\epsilon} \leq 1 - (1-\epsilon) = \epsilon$, so $\delta \leq \epsilon$ as required. For the final claim of the Lemma, note that for all $0 \leq t \leq k$, we have $\rho_k = \left(k - \frac{1}{\sqrt{1-\epsilon}}\right)\delta = \left(t - \frac{1}{\sqrt{1-\epsilon}}\right)\delta + (k-t)\delta = \rho_t + (k-t)\delta$.
\end{proof}

\subsection{Bounding on the number of swaps performed by Algorithm~\ref{alg:local-search-main}}
\label{sec:conv-algor}

In all of our preceding analysis, we have relied only on local optimality of the set $A$ produced by Algorithm~\ref{alg:local-search-main}, without considering the time required to converge to such a local optimum. Here, we show that the weight-scaling argument used by Berman~\cite{Berman:2000:d/2}, together with one round of partial enumeration, can be combined with our results to obtain a polynomial time algorithm. We first briefly review the general weight-scaling approach used in~\cite{Berman:2000:d/2}.

Suppose that any $A$ that is locally optimal with respect to the improvements considered by Algorithm~\ref{alg:local-search-main} for a weight function $w$ satisfies $\alpha w(A) \geq w(O)$ for some approximation factor $\alpha \geq 1$. Let $G(V,E)$ be a given claw-free graph with weights $w_v$ for $v \in V$, and let $O \subseteq V$ be an independent set of $G$ with maximum weight. We run the standard greedy algorithm to construct a solution $S_0$ and then set $d \triangleq \frac{n}{\epsilon w(S_0)}$. We then define a new instance of the problem using the weight function $\tilde{w}_v \triangleq \left\lfloor d w_v \right\rfloor$ for all $v \in V$ and apply Algorithm~\ref{alg:local-search-main} to this new instance, starting from the solution $S_0$. Then, for all sets $S$ maintained by Algorithm~\ref{alg:local-search-main} algorithm, we have $\tilde{w}(S) \leq d w(S) \leq d w(O)$ and since the weights $\tilde{w}_v$ are integral, the algorithm can thus make at most
\begin{equation*}
\tilde{w}^2(O) - \tilde{w}^2(S_0) \leq \tilde{w}^2(O) \leq k \tilde{w}^2(S_0) \leq k\tilde{w}(S_0)^2 \leq k\left(dw(S_0)\right)^2 = kn^2\epsilon^{-2}
\end{equation*}
improvements before arriving at a locally optimal set $A$. For the second inequality, note that whenever $w_a \leq w_b$, $\tilde{w}^2_a \leq \tilde{w}^2_b$ as well, and so any greedy solution for weight function $w$ is also greedy solution for weight function $\tilde{w}^2$. The inequality then follows since the greedy algorithm has an approximation factor of at most $k$ for the maximum weighted independent set problem in $(k+1)$-claw free graphs.

Let $A$ be the locally optimal solution produced by applying Algorithm~\ref{alg:local-search-main} to $G$ with weight function $\tilde{w}$. Then, $\alpha \tilde{w}(A) \geq \tilde{w}(O)$ and so
\begin{equation*}
 \alpha dw(A) \geq \alpha \tilde{w}(A) \geq \tilde{w}(O) \geq d w(O) - |O|,
\end{equation*}
which in turn implies
\begin{equation*}
\alpha w(A) \geq w(O) - \frac{\epsilon w(S_0)}{n} |O| \geq w(O) - \epsilon w(S_0) \geq w(O) - \epsilon w(O).
\end{equation*}
Altogether, then applying Algorithm~\ref{alg:local-search-main} to $\tilde{w}$ gives us an approximation factor of $\alpha/(1-\epsilon)$ by using at most $kn^2\epsilon^{-2}$ improvements.

We now show that in fact this loss of $\epsilon$ can be removed entirely. For each $v \in V$, we construct a residual instance $G'(V',E') = G[V \setminus N(v,V)]$. We then run the above local search routine on $G'$ with $\epsilon = (\alpha-1)n^{-1} = \Omega(n^{-1})$ and return the best solution obtained across all $n$ instances. Note that for any independent set $I$ in $G'$, $I \cup \{v\}$ is an independent set in $G$. Let $\hat{v} = \arg\max_{a \in O}w_a$ be the heaviest vertex in the optimal solution and consider the residual instance in which $v = \hat{v}$. Let $A'$ be the solution produced by our algorithm on this instance and let $O' \triangleq O - \hat{v}$. Then, $A = A' \cup \{\hat{v}\}$ is an independent set in $G$ and
\begin{multline*}
\alpha w(A) = \alpha w_{\hat{v}} + \alpha w(A')
\geq \alpha w_{\hat{v}} + w(O') - \epsilon w(O')
= w(O) + (\alpha-1)w_{\hat{v}} - \frac{\alpha-1}{n} w(O')
\geq w(O),
\end{multline*}
where the last inequality follows from $w_{\hat{v}} = \max_{v \in O}w_v \geq \frac{1}{|O|}w(O) \geq \frac{1}{n}w(O) \geq \frac{1}{n}w(O')$. Altogether then, considering the best of all $n$ solutions produced by the algorithm gives us a solution of weight at least $w(A)$ and so we obtain a factor $\alpha$ approximation. Moreover, the final algorithm performs at most $n^3k^2\epsilon^{-1} = O(n^4k)$ improvements across all $n$ iterations of the algorithm.

\subsection{A matching lower bound for the analysis in Section~\ref{sec:larger-exchanges}}
\label{sec:matching-lower-bound}
Here we give a small example to show that novel ideas have to be incorporated in order to be improve our analysis from Section~\ref{sec:larger-exchanges}. This analysis leads to a factor of $1.81$ when $k = 3$, by balancing the improvement $\frac{1-\epsilon}{2-\epsilon}$ obtained for non-isolated vertices with the improvement $\rho_3$ obtained for isolated vertices. This leads to a value $\epsilon \approx 0.3918$ (see Table~\ref{tab:ratio-full}).

The example shown in Figure \ref{fig:bad-ex} provides an almost tight example of our analysis, up to an error of $0.02$ in the approximation.
The example consists of a central vertex with $3$ vertices of $O$ mapped to it by $\pi$. We connect this central vertex by two paths of vertices all of which are connected to $2$ vertices of $O$ in the mapping $\pi$.
The weights of the vertices are set so each vertex in $A$ is isolated in $H_{\e'}$ for some $\epsilon'$ infinitesimally smaller than $\epsilon$. Thus, in our analysis, we will consider each claw as a single swap. The weights of the vertices in $\opt$ are fixed so that $\Delta_a = 0$ for all $a \in A$. Note that our example is \emph{not} a tight example for Algorithm \ref{alg:local-search-main} since there is an improving $2$-exchange. However, as we will show this example implies that to make further progress we need to either consider larger swaps involving isolated vertices, or find an improved bound $\frac{1-\e}{2-\e}$ for non-isolated vertices, allowing us to increase $\e$ in our final analysis.
\begin{figure}
  \centering
  \begin{tikzpicture}[scale=0.7, every node/.style={scale=0.75}]
      \node[shape=circle, draw, thick, label={$1$}] (1) at (0,0) {};
      \node[shape=circle, draw, thick, label={$1- \e$}] (2) at (3,0) {};
      \node[shape=circle, draw, thick, label={$(1-\e)^2$}] (3) at (6,0) {};
      \node[shape=circle, draw, thick, label={$(1- \e)^\ell$}] (4) at (9,0) {};

      \node[shape=circle, thick, draw] (a) at (1,-2) {};
      \node[shape=circle, draw, thick, label=below:{$\sqrt{\tfrac{1 + 2(1-\e)^2}{3}}$}] (b) at (0,-2) {};
      \node[shape=circle, draw, thick] (c) at (-1,-2) {};

      \node[shape=circle, draw, thick, label=below:{\small{$(1-\e)\sqrt{\tfrac{1 + (1-\e)^2}{2}}$}}] (d) at (3,-2) {};
      \node[shape=circle, draw, thick] (e) at (4,-2) {};

      \node[shape=circle, draw, thick, ] (f) at (6,-2) {};
      \node[shape=circle, draw, thick, label=below:\small{$(1-\e)^2\sqrt{\tfrac{1 + (1-\e)^2}{2}}$\phantom{blabla}}] (g) at (7,-2) {};

      \node[shape=circle, draw, thick, ] (h) at (9,-2) {};
      \node[shape=circle, draw, thick, label=below:\small{$\frac{(1-\e)^\ell}{\sqrt{2}}$}] (i) at (10,-2) {};

      \node[shape=circle, draw, thick, label={$(1-\e)$}] (-1) at (-3, 0) {};
      \node[shape=circle, draw, thick, label={$(1-\e)^2$}] (-2) at (-6, 0) {};
      \node[shape=circle, draw, thick, label={$(1-\e)^\ell$}] (-3) at (-9, 0) {};

      \node[shape=circle, draw, thick, label = below:{\small{$(1-\e)\sqrt{\tfrac{1 + (1-\e)^2}{2}}$}}] (-d) at (-3, -2) {};
      \node[shape=circle, draw, thick,] (-e) at (-4, -2) {};
      \node[shape=circle, draw, thick,] (-f) at (-6, -2) {};
      \node[shape=circle, draw, thick,label = below:{\small{\phantom{blabla}$(1-\e)^2\sqrt{\tfrac{1 + (1-\e)^2}{2}}$}}] (-g) at (-7, -2) {};
      \node[shape=circle, draw, thick, ] (-h) at (-9 , -2) {};
      \node[shape=circle, draw, thick, label = below:{\small{$\frac{(1-\e)^\ell}{\sqrt{2}}$}}] (-i) at (-10, -2) {};

      \path[->, thick] (a) edge node {} (1);
      \draw[->, thick] (b) -- (1);
      \draw[->, thick] (c) -- (1);
      \draw[thick] (a) -- (2);

      \draw[->, thick] (d) -- (2) ;
      \draw[->, thick] (e) -- (2);
      \draw[thick] (e) -- (3);

      \draw[->, thick] (f) -- (3);
      \draw[->, thick] (g) -- (3);

      \draw[->, thick] (h) -- (4);
      \draw[->, thick] (i) -- (4);
      \draw[loosely dotted, thick] (g) -- (4);

      \begin{scope}[
          every edge/.style={draw=black, thick}]
          \path[->] (c) edge node {} (-1);
          \path[->] (-d) edge node {} (-1);
          \path[->] (-e) edge node {} (-1);
          \path (-e) edge node {} (-2);
          \path[->] (-f) edge node {} (-2);
          \path[->] (-g) edge node {} (-2);
          \path[loosely dotted] (-g) edge node {} (-3);
          \path[->] (-h) edge node {} (-3);
          \path[->] (-i) edge node {} (-3);

      \end{scope}
    \end{tikzpicture}
    \caption{Almost tight example for our analysis, where the vertices at the top are the vertices in the current solution, and vertices at the bottom are the vertices in the optimal solution. The values written are for individual vertex.}
    \label{fig:bad-ex}
\end{figure}
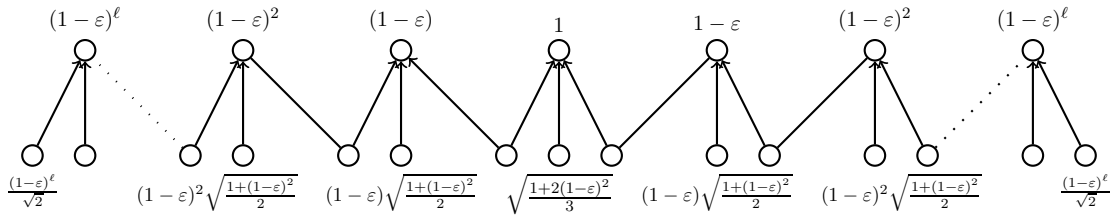
The approximation ratio of tight example is equal to:
    \begin{align*}
        \frac{w(A)}{w(O)} &  = \frac{ 1 + 2\sum_{i = 1}^\ell (1-\e)^{i}}{3\sqrt{\tfrac{1 + 2(1-\e)^2}{3}} + 4\sqrt{\tfrac{1 + (1-\e)^2}{2}} \sum_{i = 1}^{\ell-1} (1-\e)^{i} + 4\frac{(1-\e)^\ell}{\sqrt{2}}}, \\
        & \underset{\ell \rightarrow \infty}{\rightarrow} \frac{2 \e^{-1} - 1}{3\sqrt{\tfrac{1 + 2(1-\e)^2}{3}} + 4\sqrt{\tfrac{1 + (1-\e)^2}{2}}(\e^{-1} - 1)}.
    \end{align*}
For $\e = 0.3918$, the value of the previous ratio is equal to $\simeq (1.80857)^{-1}$. In contrast, for the same value of $\e$, the bound obtained for non-isolated vertices is equal to $(2 - \frac{1 - \e}{2(2- \e)})^{-1} = (1.81091)^{-1}$. 

    Figure \ref{fig:bad-ex} demonstrates that minor modifications of our current analysis cannot beat a factor of $1.8$. This example captures the tension that the variable $\e$ faces. On the one hand, the approximation factor of Figure \ref{fig:bad-ex} decreases as $\e$ increases. But, as $\e$ increases the bound for the exchange, i.e., $\frac{1 - \e}{2-\e}$, decreases. This suggests that to surpass the $1.8$ factor, we must either improve our bound for non-isolated vertices, or extend our techniques to combine isolated vertices into multiple swaps.
